\documentclass[11pt]{article}

\usepackage{fullpage}
\usepackage{amsmath, amsthm, amssymb}
\usepackage{bbm}
\usepackage{mathtools}
\usepackage{graphicx}
\usepackage{xcolor}
\usepackage{enumerate}
\usepackage{natbib}
\usepackage{subfig}
\usepackage{multirow}
\usepackage{float}
\usepackage{tikz}
\usepackage{longtable}
\usepackage{indentfirst}
\usepackage{setspace}
\usepackage{xurl}
\RequirePackage[colorlinks,citecolor=blue,urlcolor=blue]{hyperref}

\newtheorem{lemma}{Lemma}
\newtheorem{cor}{Corollary}
\newtheorem{theorem}{Theorem}
\newtheorem{example}{Example}
\newtheorem{prop}{Proposition}

\DeclareMathOperator{\diag}{diag}
\DeclareMathOperator{\Cov}{Cov}
\DeclareMathOperator{\Tr}{Tr}

\DeclareMathOperator{\MISE}{MISE}
\DeclareMathOperator{\AMISE}{AMISE}
\DeclareMathOperator{\AIC}{AIC}
\DeclareMathOperator{\IB}{IB}

\definecolor{c1}{rgb}{0,  0, 1}
\definecolor{c2}{rgb}{1,  0, 0}
\definecolor{c3}{rgb}{128,  0, 128}

\DeclareRobustCommand \triangled{\tikz{\filldraw[color =  c3,fill= c3] (0,0.1) -- (0.1,0.3) --(0.2,0.1) -- cycle;}}

\DeclareRobustCommand\full  {\tikz[baseline=-0.6ex]\draw[c1, line width=0.35mm] (0,0)--(0.57,0);} 
\DeclareRobustCommand\denselydashed{\tikz[baseline=-0.6ex]\draw[c2,line width=0.35mm, dash pattern={on 7pt off 1.5pt}] (0,0)--(0.57,0);} 

\begin{document}

\title{Robust and efficient estimation of nonparametric generalized linear models}
\author{Ioannis Kalogridis$^1$, Gerda Claeskens$^2$ and Stefan Van Aelst$^1$}
\date{%
    $^1$ Department of Mathematics, KU Leuven  \\%
    $^2$ ORStat and Leuven Statistics Research Centre, KU Leuven \\[2ex]%
    \today
}
\maketitle

\begin{abstract}
Generalized linear models are flexible tools for the analysis of diverse datasets, but the classical formulation requires that the parametric component is correctly specified and the data contain no atypical observations. To address these shortcomings we introduce and study a family of nonparametric full rank and lower rank spline estimators that result from the minimization of a penalized power divergence. The proposed class of estimators is easily implementable, offers high protection against outlying observations and can be tuned for arbitrarily high efficiency in the case of clean data. We show that under weak assumptions these estimators converge at a fast rate and illustrate their highly competitive performance on a simulation study and two real-data examples.
\end{abstract}



\section{Introduction}

Based on data $(t_1, Y_1), \ldots, (t_n, Y_n)$ with fixed $t_i \in [0,1]$, consider the generalized linear model (GLM) stipulating that
\begin{align}
Y_i \sim F_{\theta_{0,i}, \phi_0}, \quad (i=1, \ldots, n),
\label{GLM}
\end{align}
where $F_{\theta_{0,i}, \phi_0}$ is a discrete or continuous exponential family of distributions over $\mathbbm{R}$. Here, $\theta_{0,i}$ is the ``canonical" parameter depending on $t_i$ and $\phi_0$ is a common nuisance or scale parameter. Each $F_{\theta_{0,i}, \phi_{0}}$ has density with respect to either the Lebesgue or counting measure of the form
\begin{align}
\label{eq:2}
f_{\theta_{0,i}, \phi_0}(y) = \exp\left[ \frac{y \theta_{0,i} - b(\theta_{0,i})}{\phi_0} + c(y, \phi_0) \right], \quad (i=1, \ldots, n),
\end{align}
for some functions $b: \mathbbm{R} \to \mathbbm{R}$ and $c: \mathbbm{R} \times \mathbbm{R}_{+} \to \mathbbm{R} $ that determine the shape of the density. In the classical formulation \citep[see, e.g.,][]{MCN:1983}, the systematic component of the model is defined indirectly through the relation $G(b^{\prime}(\theta_{0,i})) = \beta_0 + \beta_1 t_i$ for a known link function $G: \mathbbm{R} \to \mathbbm{R}$ and some unknown $(\beta_0,\beta_1) \in \mathbbm{R}^2$. In the present paper, following \citet{Green:1994}, we avoid this heavy parametric assumption and simply require that $\theta_{0,i} = g_0(t_i)$ for a smooth function $g_0: [0,1] \to \mathbbm{R}$ to be estimated from the data.

The limitations of parametric inference in GLMs have been previously noticed by several authors and a large number of broader alternatives have been proposed through the years. The interested reader is referred to the dedicated monographs of \citet{Green:1994}, \citet{Ruppert:2003} and \citet{Gu:2013} for extensive relevant discussions and illustrative examples of nonparametric GLM estimators. From a practical standpoint, an important drawback of most of these methods is their reliance on a correctly specified probability model for all the data, whereas, as \citet[p.\,3]{Huber:2009} note, mixture distributions and aberrant observations not following the model frequently occur in practice. The need for resistant estimators in the nonparametric framework has been acknowledged as early as \citet{Hastie:1990}, but the literature in the meantime has remained relatively sparse and overwhelmingly focused on continuous responses with constant variance. For this type of response, resistant penalized spline estimators have been considered by \citet{Kalogridis:2021} and \citet{Kal:2021}, but unfortunately these ideas do not naturally extend to the setting of GLMs.

A resistant estimation procedure encompassing many generalized linear models was proposed by \citet{Ali:2011}, who devised a robust backfitting algorithm based on the robust quasi-likelihood estimator of \citet{Cantoni:2001b} along with local polynomial estimation and showed the asymptotic unbiasedness of their estimates under a suitable set of conditions. \citet{Bianco:2011} investigated general weighted M-estimators defined and proved their pointwise consistency and asymptotic normality. \citet{Croux:2012} combined the approach of \citet{Cantoni:2001b} with the P-spline approach of \citet{Eilers:1996} in order to construct robust estimates of both mean and dispersion, but without theoretical support. For smoothing-spline type robust quasi-likelihood estimators, \citet{Wong:2014} showed that for smooth score functions and under certain limit conditions the robust estimator inherits the rate of convergence of the non-robust quasi-likelihood estimator. More recently, \citet{Aeb:2021} proposed obtaining robust estimates by wrapping log-likelihood contributions with a smooth convex function that downweights small log-likelihood contributions, but again with limited theoretical support.

A drawback shared by all aforementioned approaches is the lack of automatic Fisher consistency, that is, these methods do not estimate the target quantities without ad-hoc corrections. These corrections may seem unnatural to researchers and practitioners that are well-acquainted with classical maximum likelihood estimators. Moreover, the estimators of \citet{Ali:2011}, \citet{Croux:2012} and \citet{Wong:2014} depend on monotone score functions, which implies that outliers in both response and predictor spaces can still exert undue influence on the estimates. \citet{Croux:2012} proposed using a weight function to limit the influence of high leverage observations in the predictor space, but this greatly complicates theoretical investigations and does not help with gross outliers in the response space. On the other hand, the pointwise nature of the estimators of \citet{Bianco:2011} poses a computational challenge and makes uniform asymptotic results much more difficult to obtain.

To overcome these issues, we propose a new family of non-parametric spline estimators for GLMs that is based on the concept of density power divergence between distributions, as developed by \citet{Basu:1998}. We propose both a full-rank smoothing spline type estimator and a lower-rank estimator based on the principle of penalized splines, which is particularly advantageous for large datasets. The proposed estimators possess a number of notable advantages. First, they are inherently Fisher-consistent for all GLMs and thus do not require ad-hoc corrections. The proposed class of estimators is adaptive and can combine high efficiency in clean data with robustness against gross outliers.  Moreover, these estimators can be efficiently computed through combination of the locally supported B-splines and fast iterative algorithms. For this family of estimators we establish asymptotic rates of convergence in a Sobolev norm, a rate of uniform convergence and rates of convergence of the derivatives, the latter two of which are novel outside the classical nonparametric likelihood framework.

\section{The proposed family of estimators}
\label{sec:2}

\subsection{Density power divergence}

We begin by reviewing the definition of density power divergence, as introduced by \citet{Basu:1998} for i.i.d. data and modified by \citet{Ghosh:2013} for independent but not identically distributed data. Consider two densities $f$ and $h$, which for clarity we take to be Lebesgue-densities. Densities with respect to the counting measure are also covered by the subsequent arguments provided essentially that integrals are replaced by sums. The density power divergence $d_{\alpha}(h,f)$ is defined as
\begin{align*}
d_{\alpha}(h, f) = \begin{dcases} \int_{\mathbbm{R}} \left\{ f^{1+\alpha}(x) - \left(1+\frac{1}{\alpha} \right)h(x) f^{\alpha}(x) + \frac{1}{\alpha}h^{1+\alpha}(x) \right\}dx & \alpha>0 \\ \int_{\mathbbm{R}} h(x) \log\{h(x)/f(x) \} dx & \alpha = 0.  \end{dcases}
\end{align*}
It is easy to see that $d_{\alpha}(h, f)  \geq 0$ and that $d_{\alpha}(h, f) = 0$ if and only if $h = f$ almost everywhere. Moreover, it can be readily verified that $d_{\alpha}(h, f)$ is continuous for $\alpha \to 0$. In fact, $d_0(h, f)$ is the Kullback-Leibler divergence, which is closely associated with classical maximum likelihood estimators \citep[see, e.g.,][]{Claeskens:2008}, whereas $d_1(h, f)$ is the $L_2$-error between densities, which has been used by \citet{Scott:2001} for robust parametric modelling. Hence, as \citet{Basu:1998} note, density power divergence provides a smooth bridge between maximum likelihood estimation and $L_2$-distance minimization.

In the GLM setting we assume that the densities of the $Y_i$ share a possibly infinite-dimensional parameter $\theta_0$ determining the form of their conditional means through the relationship $G(\mu_i) = \ell_i (\theta_0)$ where $\ell_i$ is the evaluation functional for each $t_i$, i.e., $\ell_i (\theta_0) = \theta_0(t_i)$. Moreover, $\theta_0$ belongs to some metric space $\Theta$, which is known a priori. Henceforth, we shall place emphasis on estimating $\theta_{0}$, and assume that $\phi_0$, the dispersion parameter, is either known or suitably substituted. For popular GLMs such as logistic, Poisson and exponential models, $\phi_0$ is indeed known whereas for other GLMs, e.g., with Gaussian responses, $\phi_0$ can be estimated without any model fitting with the resistant Rice-type estimators proposed by \citet{Ghement:2008} and \citet{Boente:2010}. Therefore, we set $\phi_0 = 1$ without loss of generality and drop it from the notation from now on. Furthermore, for convenience we shall from now on use $\theta_i$ to denote $\ell_i(\theta)$. 

To derive estimators from the density power divergence, we fix an $\alpha>0$ and minimize $n^{-1} \sum_{i=1}^n d_{\alpha}(f_{\theta_{0,i}}, f_{\theta_i})$ over all $\theta \in \Theta$. For this minimization problem we may drop terms only involving $f_{\theta_{0,i}}$ as they do not depend on $\theta$. An unbiased estimator for the unknown cross term  $\int f_{\theta_i}^{\alpha}d F_{\theta_{0,i}}$ is given by $f_{\theta_i}^{\alpha}(Y_i)$. Thus, we may replace each summand by
\begin{align}
\label{eq:3}
l_{\alpha}(Y_i, \theta_i) = \int_{\mathbbm{R}} f_{\theta_i}^{1+\alpha}(x) dx - \left(1+\frac{1}{\alpha} \right) f_{\theta_i}^{\alpha}(Y_i), \quad (i =1 \ldots, n).
\end{align}
Below we explain how we use this principle to construct a new class of nonparametric GLM estimators with good asymptotic properties.

In the remainder of the paper we assume that $\Theta$ is a reproducing kernel Hilbert space (RKHS) of functions that is generated by a kernel $\mathcal{R}:[0,1]^2 \to \mathbbm{R}$. Let us denote this space by $\mathcal{H}(\mathcal{R})$. A crucial property of $\mathcal{H}(\mathcal{R})$ is that for every $f \in \mathcal{H}(\mathcal{R})$ it holds that $f(t_i) = \langle f, \mathcal{R}(t_i,\cdot) \rangle_{\mathcal{H}(\mathcal{R})}$.  A popular measure of robustness is the influence function (IF), which roughly measures the effect of a small proportion of contamination on the estimator \citep[see, e.g.,][]{Hampel:2011}. Estimators with bounded IF are considered robust, as in that case a small amount of contamination can only have a limited effect on the estimator. 

It can be shown that for every $\alpha>0$ the associated density power divergence estimator possesses a bounded influence function. Indeed, using standard M-estimation theory, it may be verified that the influence function for the density power divergence functional is proportional to
\begin{align}
\label{IF}
\sum_{i=1}^n \left[u_{\theta_{0,i}}(y_0) f^{\alpha}_{\theta_{0,i}}(y_0) - \int_{\mathbbm{R}} u_{\theta_{0,i}}(y) f_{\theta_{0,i}}^{1+\alpha}(y) dy   \right] \langle \mathcal{R}(t_i, \cdot), \cdot \rangle_{\mathcal{H}(\mathcal{R})}
\end{align}
where $u_{\theta}(y) = (y - b^{\prime}(\theta))$ is the derivative of the log-density with respect to its canonical parameter and $u_{\theta_{0,i}}(y)$ is the derivative evaluated at $\theta_{0,i}$. An analogous result holds for the discrete case. For $\alpha=0$ the IF of the maximum likelihood estimator is obtained, which is clearly an unbounded function of $y_0$. However, for $\alpha>0$ it may be seen that the IF is bounded, ensuring the resistance of density power divergence estimators are robust against small amounts of contamination. Their degree of resistance depends on the magnitude of $\alpha$, as larger values of $\alpha$ lead to the  faster decay of $u_{\theta}(y_0) f^{\alpha}_{\theta}(y_0)$ to zero, ensuring greater robustness.

\subsection{Smoothing spline type estimators}

Consider now the specific GLM~\eqref{GLM} with densities~\eqref{eq:2} where $\theta_{0,i} = \theta_0(t_i) = g_0(t_i)$ for $i=1, \ldots, n$, and $g_0$ needs to be estimated from the data. In this section we only require that $g_0$ belongs to the Hilbert-Sobolev space $\mathcal{W}^{m,2}([0,1])$ for some $m \geq 1$, which is defined as
\begin{align*}
\mathcal{W}^{m,2}([0,1]) = \{ f:[0,1] \to \mathbb{R}, f\ &\text{has $m-1$ absolutely continuous derivatives}  \nonumber \\  & f^{(1)}, \ldots, f^{(m-1)}\ \text{and} \ \int_0^1  \vert f^{(m)}(t) \vert^2 dt < \infty  \} \nonumber.
\end{align*}
The space $\mathcal{W}^{m,2}([0,1])$ is well-suited for nonparametric regression problems, as it forms a RKHS so that for each $x \in [0,1]$ the evaluation functionals $\mathcal{W}^{m,2}([0,1]) \to \mathbbm{R}: f \mapsto f(x) $ are continuous, see, e.g., \citet{Wahba:1990}.

As a compromise between goodness of fit and complexity we propose to estimate $g_0$ by the function $\widehat{g}_n$ solving
\begin{equation}
\inf_{g \in \mathcal{W}^{m,2}([0,1])} \left[ \frac{1}{n} \sum_{i=1}^n l_{\widehat{\alpha}_n}(Y_i,g(t_i)) + \lambda \int_{[0,1]} \vert g^{(m)}(t) \vert^2 dt \right],
\label{eq:5}
\end{equation}
for some $\lambda>0$, that acts as the tuning parameter. Here, we also allow for a random tuning parameter $\widehat{\alpha}_n$. This random tuning parameter may depend on the data itself leading to an adaptive estimator whose robustness and efficiency automatically adjust to the data. In particular, for $\widehat{\alpha}_n$ close to zero the objective function approaches the  penalized likelihood considered, for example, by  \citet{Cox:1990}, \citet{Mammen:1997} and \citet{Kau:2009}. As discussed previously, these estimators are efficient but not robust. On the other hand, for large $\widehat{\alpha}_n$, estimators minimizing \eqref{eq:5} are robust but not efficient. In practice, we aim to balance robustness and efficiency and select an $\widehat{\alpha}_n$ in $(0,1]$, although higher values can also be considered. Section~\ref{sec:4} outlines a possible strategy in this respect.

For bounded densities that are continuous with respect to their parameter, the objective function is bounded from below and continuous in $\mathcal{W}^{m,2}([0,1])$. Reasoning along the same lines as in the proof of Theorem 1 of \citet{Kalogridis:2021} reveals that for $n \geq m$ this minimization problem is well-defined and there exists at least one minimizer in $\mathcal{W}^{m,2}([0,1])$.  Arguing now in a standard way \citep[see, e.g.,][]{Eubank:1999} shows that this minimizer must be an easily computable $n$-dimensional natural spline with knots at the unique $t_1, \ldots, t_n$. As we discuss in Section~\ref{sec:4} below though, unrestricted B-splines may also be used in the computation of the estimator.

Even for Gaussian responses the smoothing spline type estimator in \eqref{eq:5} has not been previously considered. In this case the form of the loss function is rather simple. Indeed, for Gaussian $Y_i$ the first term in \eqref{eq:5} is constant as a function of $g \in \mathcal{W}^{m,2}([0,1])$. Hence, apart from constants the objective function becomes
\begin{align*}
\frac{1}{n} \sum_{i=1}^n \rho_{\widehat{\alpha}_n}(Y_i-g(t_i)) + \lambda \int_{[0,1]} \vert g^{(m)}(t) \vert^2 dt,
\end{align*}
with $\rho_{\alpha}(x) = -e^{- \alpha x^2/2}$. This exponential loss function has attractive properties for robust estimation because it is a bounded loss function which is infinitely differentiable with bounded derivatives of all orders. In the parametric setting the exponential squared loss function has been used, e.g., by \citet{Wang:2013}. In the nonparametric setting considered herein, the penalized exponential squared loss gives rise to a novel estimator that may be viewed as a more robust alternative to the Huber and least absolute deviations smoothing spline estimators studied in \citet{van de Geer:2000, Kalogridis:2021}. See Section~\ref{sec:5} for interesting comparisons.

A noteworthy property of the penalty functional in \eqref{eq:5} is the shrinkage of the estimator towards a polynomial of order $m$. To see this, assume that $\widehat{g}_n$ lies in the null space of the penalty so that $\|\widehat{g}_n^{(m)}\| = 0$. A Taylor expansion with integral remainder term shows that
\begin{align*}
\widehat{g}_n(t) = P_m(t) + \int_0^1 \frac{\widehat{g}_n^{(m)}(x)}{(m-1)!} (t-x)^{m-1}_{+} dx,
\end{align*}
where $P_{m}(t)$ is the Taylor polynomial of order $m$. The Schwarz inequality shows that the integral remainder vanishes for all $t \in [0,1]$, whence $ \sup_{t \in [0,1]} \vert \widehat{g}_n (t) - P_{m} (t) \vert= 0$. This implies that letting $\lambda \to \infty$ will cause the estimator to become a polynomial of order $m$, as the dominance of the penalty term in \eqref{eq:5} forces the estimator to lie in its null space.

A crucial property underlying the construction of all our GLM estimators is their inherent Fisher-consistency \citep[p.~83]{Hampel:2011}. In particular, our previous discussion shows that, for each $\alpha>0$, $\theta_{0,i}$ minimizes $\mathbb{E}\{l_\alpha(Y_i, \theta_i) \}$ 
for each $i = 1, \ldots, n$. Hence, our estimation method is Fisher-consistent for every model distribution. To the best of our knowledge, this is the first robust nonparametric estimator that enjoys this property without corrections, although other divergence based-estimators may also enjoy this property. Since Fisher-consistency corrections are model-dependent, the inherent Fisher-consistency yields an important practical bonus.

\subsection{Penalized spline type alternatives}

A drawback of smoothing spline type estimators is their dimension, which grows linearly with the sample size. This implies that for large $n$ smoothing spline estimators can be computationally cumbersome. Moreover, as noted by \citet[p.199]{Wood:2017}, in practice the value of $\lambda$ is almost always high enough such that the effective degrees of freedom of the resulting spline is much smaller than $n$. Penalized spline estimators offer a compromise between the complexity of smoothing splines and the simplicity of (unpenalized) regression splines \citep{O:1986, Eilers:1996}. We now discuss penalized spline estimators for our setting as a simpler alternative to the smoothing spline estimator discussed above.

Fix a value $K \in \mathbbm{N}_{+}$ and define the interior knots $0=x_0 <x_1, \ldots, x_{K} < x_{K+1} = 1$,  which do not have to be design points. For a fixed $p\in \mathbbm{N}_{+}$, let $S_{K}^p([0,1])$ denote the set of spline functions on $[0,1]$ of order $p$ with knots at the $x_i$. For $p=1,\ S_{K}^1([0,1])$ is the set of step functions with jumps at the knots while for $p\geq 2$,
\begin{align*}
S_{K}^p([0,1]) = \{ s \in \mathcal{C}^{p-2}([0,1]): s(x)\ &\text{is a polynomial of degree $(p-1)$ } \\ &
\text{on each $[x_i, x_{i+1} ]$} \}.
\end{align*}
Thus, $p$ controls the smoothness of the functions in $S_{K}^p$ while the number of interior knots $K$ represents the degree of flexibility of spline functions in $S_{K}^p([0,1])$, see e.g., \citet{Ruppert:2003,Wood:2017}. It is easy to see that $S_{K}^p([0,1])$ is a $(K+p)$-dimensional subspace of $\mathcal{C}^{p-2}([0,1])$ and B-spline functions yield 
a stable basis for $S_{K}^p([0,1])$ with good numerical properties \citep{DB:2001}.

For any $m \in \mathbbm{N}_{+}$ satisfying $m<p$ we define the penalized spline type estimator $\widehat{g}_n$ as the solution of the optimization problem
\begin{align}
\label{eq:6}
\min_{g \in S_{K}^p([0,1])} \left[ \frac{1}{n} \sum_{i=1}^n l_{\widehat{\alpha}_n}(Y_i,g(t_i)) + \lambda \int_{[0,1]} \vert g^{(m)}(t)\vert^2 dt\right],
\end{align}
with $\lambda \geq 0$.  Hence, we have replaced $\mathcal{W}^{m,2}([0,1])$ in~\eqref{eq:5} by a $(K+p)$- dimensional spline subspace. For $K \ll n$ this yields sizeable computational gains in relation to the smoothing spline estimator. Moreover, it turns out that penalized spline estimators do not sacrifice much in terms of accuracy if $K$ is large enough, but still smaller than $n$. See \citet{Claeskens:2009, Xiao:2019} for interesting comparisons in classical nonparametric regression models and Section~\ref{sec:3} below for a comparison in the present context.

Penalized spline estimators retain a number of important mathematical properties of their full rank smoothing spline counterparts. In particular, for $\lambda>0$ and $p = 2m$ it can be shown that the penalized spline estimator is a natural spline of order $2m$. Moreover, the null space of the penalty consists exactly of polynomials of order $\leq m$. In the frequently used setting of equidistant interior knots, the latter property is also retained if one replaces the derivative penalty with the simpler difference (P-spline) penalty $\sum_{j = m+1}^{K+p} \vert \Delta^m \beta_j \vert^2$ proposed by \citet{Eilers:1996}. Here, $\Delta^m$ is the $m$th backward difference operator and $\beta_j, j = 1, \ldots, K+p$ are the coefficients of the B-spline functions. In this case, these two penalties are scaled versions of one another with the scaling factor depending on $K$, $p$ and $m$ \citep[see, e.g., Proposition 1 of][]{Kal:2021}. Thus, P-spline estimators are also covered by the asymptotic results of the following section.

\section{Asymptotic behaviour of the estimators}
\label{sec:3}

\subsection{Smoothing spline type estimators}

As noted before, an essential characteristic of $\mathcal{W}^{m,2}([0,1])$
is that it is a RKHS. The reproducing kernel depends on the inner product that $\mathcal{W}^{m,2}([0,1])$ is endowed with. We shall make use of the inner product
\begin{align*}
\langle f, g \rangle_{m, \lambda} = \langle f, g \rangle + \lambda \langle f^{(m)}, g^{(m)} \rangle, \quad f,g \in \mathcal{W}^{m,2}([0,1])
\end{align*}
where $\langle \cdot, \cdot \rangle$ denotes the standard inner product on $\mathcal{L}^2([0,1])$.  It is interesting to observe that $\langle \cdot, \cdot \rangle_{m, \lambda}$ is well-defined and depends on the smoothing parameter $\lambda$, which typically varies with $n$. \citet[Chapter 13]{Eg:2009} show that there exists a finite positive constant $c_0$ such that for all $f \in \mathcal{W}^{m,2}([0,1])$ and $\lambda \in (0, 1]$ we have
\begin{align}
\label{eq:7}
\sup_{t\in [0,1]} |f(t)| \leq c_0 \lambda^{-1/(4m)} \left\|f\right\|_{m, \lambda},
\end{align}
with $\left\|f\right\|_{m, \lambda} = \langle f, f \rangle_{m,\lambda}^{1/2}$.
Hence, for any $\lambda \in (0,1]$, $\mathcal{W}^{m,2}([0,1])$ is indeed a RKHS under $\langle \cdot, \cdot \rangle_{m, \lambda}$. The condition $\lambda \leq 1$ is not restrictive, as for our main result below we assume that $\lambda \to 0 $ as $n \to \infty$ and our results are asymptotic in nature. 

The assumptions needed for our theoretical development are as follows.
\begin{itemize}
\item[(A1)] The support $\mathcal{Y} := \overline{ \{ y: f_{{\theta}}(y)>0 \} }$ does not depend on $\theta \in \mathbbm{R}$.
\item[(A2)] There exists an $\alpha_0>0$ such that $\widehat{\alpha}_n \xrightarrow{P} \alpha_0$.
\item[(A3)] The densities $\{f_{\theta}(y), \theta \in \mathbbm{R}\}$ are uniformly bounded, twice differentiable as a function of their canonical parameter $\theta$ in a neighbourhood of the true parameter $\theta_0$ and there exist $\delta>0$ and $M< \infty$ such that
\begin{align*} 
\sup_{t \in [0,1]} \sup_{\substack{|\alpha-\alpha_0|<\delta \\ |u| < \delta}} \sup_{y \in \mathcal{Y}} \left| \frac{\partial f_{\theta}^{\alpha}(y)}{\partial \theta} \bigg \vert_{\theta = \theta_{0}(t)+u}  \log\left(f_{\theta_0(t)+u}(y) \right) \right|  & \leq M.
\end{align*}
\item[(A4)] In the case of densities w.r.t. Lebesgue measure, the families of functions $\{m_t(u, \alpha), t \in [0,1] \}$ and $\{n_t(u, \alpha, y), t \in [0,1]\}$ defined by
\begin{align*}
m_t(u, \alpha) = \int_{\mathcal{Y}} \frac{\partial^2 f^{1+\alpha}_{\theta}(y)}{\partial \theta^2} \bigg \vert_{\theta = \theta_{0}(t)+u} dy, \quad n_t(u, \alpha, y) =  \frac{\partial^2 f^{\alpha}_{\theta}(y)}{\partial \theta^2} \bigg \vert_{\theta = \theta_{0}(t)+u},
\end{align*}
are equicontinuous at $u = 0$, for every $\alpha \in (\alpha_0-\delta, \alpha_0+\delta)$ with $\delta$ as in (A3) and $y \in \mathcal{Y}$. Moreover, there exists an $M^{\prime}>0$ such that
\begin{align*}
\sup_{t \in [0,1]}  \sup_{\substack{ |\alpha-\alpha_0|<\delta \\ |u|<\delta }} \left[ \left|m_t(u, \alpha)\right| + \sup_{y \in \mathcal{Y}} \left|n_t(u, \alpha, y)\right|  \right] \leq M^{\prime}.
\end{align*}
For densities w.r.t. counting measure, the sum over $y \in \mathcal{Y}$ replaces the integral in the definition of $m_t(u, \alpha)$.
\item[(A5)] For $\delta>0$, as given in (A3), there exist $0<c_0 \leq C_0<\infty$ such that
\begin{align*}
c_0 \leq \inf_{n} \min_{i \leq n} \mathbb{E}\{ f_{\theta_{0,i}}^{\alpha}(Y_i) |u_{\theta_{0,i}}(Y_i)|^2 \} \leq  \sup_{n}  \max_{ i \leq n} \mathbb{E}\{ f_{\theta_{0,i}}^{\alpha}(Y_i) |u_{\theta_{0,i}}(Y_i)|^2 \} \leq C_0,
\end{align*}
for all $\alpha \in (\alpha_0-\delta, \alpha_0+\delta)$,
with $u_{\theta_{0,i}}(y) = \partial \log(f_{\theta}(y))/\partial \theta$ the derivative of the log-density evaluated at $\theta_{0,i}$.
\item[(A6)] The family of design points $t_i$ is asymptotically quasi-uniform in the sense of \citet{Eg:2009}, that is, there exists an $n_0 \in \mathbbm{N}$ such that,
\begin{align*}
\sup_{f \in \mathcal{W}^{1,1}([0,1])} \frac{\left|n^{-1} \sum_{i=1}^n f(t_i) - \int_{0}^1 f(t) dt\right|}{\int_{0}^1|f^{\prime}(t)| dt }  = O(n^{-1}),
\end{align*}
for all $n \geq n_0$.
\end{itemize}

Condition (A1) is standard in the theory of exponential families and may be viewed as an identifiability condition. It is worth noting that  imposing the convergence in probability of nuisance (estimated) parameters as in (A2) is a common way of dealing with them theoretically, see, e.g., \citep[theorems 5.31 and 5.55]{VDV:1998}. Furthermore, since assumption (A2) requires no rate of convergence whatsoever, the methodology employed herein can be used to weaken such assumptions in other works, e.g., in \citet[Theorem 3]{Kalogridis:2021} whose assumptions require a specific rate of convergence of the preliminary scale estimator in the context of robust nonparametric regression with continuous responses.

Assumptions (A3)--(A4) are regularity conditions, which are largely reminiscent of classical maximum likelihood conditions \citep[cf. Theorem 5.41 in][]{VDV:1998}. These conditions essentially impose some regularity of the first and second derivatives in a neighbourhood of the true parameter. Since for any $\alpha>0$,
\begin{align*}
\frac{\partial^2 f_{\theta}^{\alpha}(y)}{\partial \theta^2} & = \alpha^2 |y-b^{\prime}(\theta)|^2 f_{\theta}^{\alpha}(y) - \alpha b^{\prime \prime}(\theta)f_{\theta}^{\alpha}(y),
\end{align*}
conditions (A3) and (A4) are satisfied for a wide variety of GLMs due to the rapid decay of $f_{\theta}^{\alpha}(y)$ for large values of $|y|$. In particular, they are satisfied for  the Gaussian, logistic and Poisson models. Similarly condition (A5) is an extension of the assumptions underpinning classical maximum likelihood estimation. For $\alpha = 0$ these moment conditions entail that the Fisher information is strictly positive and finite. The following examples demonstrate that condition (A5) holds for popular GLMs.

\begin{example}[Gaussian responses]
\normalfont
Clearly, $u_{\theta_{0,i}} = y - b^{\prime}(\theta_{0,i}) = y - \theta_{0,i}$ in this case and
\begin{align*}
 \int_{\mathbbm{R}} f_{\theta_{0,i}}^{1+\alpha} (y)|u_{\theta_{0,i}}(y)|^2 dy =  \frac{1}{\{(2\pi)^{\alpha}(1+\alpha)\}^{1/2}},
\end{align*}
so that (A5) is satisfied without any additional conditions.
\end{example}

\begin{example}[Binary responses]
\normalfont
The canonical parameter is $\theta_i = \log(p_i/(1-p_i))$ with $p_i$ denoting the probability of success for the $i$th trial. Thus,
\begin{align*}
\sum_{y \in \{0,1\}} f_{\theta_{0,i}}^{1+\alpha}(y) |u_{\theta_{0,i}}(y)|^2 = (1-p_{0,i})^{1+\alpha} p_{0,i}^2 + p_{0,i}^{1+\alpha}(1-p_{0,i})^2,
\end{align*}
where $p_{0,i} = 1/(1+e^{-\theta_{0,i}})$. Thus, (A5) is satisfied whenever $p_{0,i}$ is bounded away from zero and one, precisely as required by \citet{Cox:1990} and \citet{Mammen:1997} for classical nonparametric logistic regression.
\end{example}

\begin{example}[Exponential responses]
\normalfont
Let $f_{\theta_{0,i}}(y) =\exp[y\theta_{0,i} - \log(-\theta_{0,i})] \mathcal{I}_{(0,\infty)}(y)$ with $\theta_{0,i} = - \lambda_{0,i}$. A lengthy calculation shows that
\begin{align*}
\int_{\mathbbm{R}} f_{\theta_0,i}^{1+\alpha}(y) |u_{\theta_{0,i}}(y)|^2 dy = \frac{1+\alpha^2}{(1+\alpha)^3} \frac{1}{\lambda_{0,i}^{2-\alpha}}.
\end{align*}
Hence, (A5) is satisfied provided that $\lambda_{0,i}$ stays away from zero, since by compactness and continuity, $\lambda_0(t)$ is always bounded.

\end{example}
\begin{example}[Poisson responses]
\normalfont
Here, for positive rate parameters  $\lambda_{0,i}$, we have $\theta_{0,i} = \log (\lambda_{0,i}), i = 1, \ldots, n$. It can be shown that
\begin{align*}
\lambda_{0,i}^2 \exp[-\lambda_{0,i}(1+\alpha)] \leq \sum_{y = 0}^{\infty} f_{\theta_0,i}^{1+\alpha}(y) |u_{\theta_{0,i}}(y)|^2 \leq \lambda_{0,i},
\end{align*}
so that (A5) is satisfied provided that $\lambda_{0,i}$ stays away from zero.
\end{example}

It is easy to see that assumptions (A1)--(A5) also cover density power divergence estimators based on a fixed $\alpha_0>0$, for example, $L_2$-distance estimators with $\alpha_0=1$, as in this case condition (A2) is trivially satisfied and conditions (A3)--(A5) only need to hold for that particular $\alpha_0$. Finally, condition (A6) ensures that the design points are well-spread throughout the interval of interest. Call $Q_n$ the distribution function of $t_1, \ldots, t_n$. Then, for each $f \in \mathcal{W}^{1,1}([0,1])$, an integration by parts argument reveals that
\begin{align*}
\left|n^{-1} \sum_{i=1}^n f(t_i) - \int_{0}^1 f(t) dt\right| \leq \sup_{t \in [0,1]}|Q_n(t) - t| \int_{0}^1 |f^{\prime}(t)| dt.
\end{align*}
Consequently, (A6) is satisfied provided that $Q_n$ approximates well the uniform distribution function, for example if $t_i = i/(n+1)$ or $t_i = 2i/(2n+1)$. 

Theorem~\ref{thm:1} contains the main result for smoothing type spline estimators defined in~\eqref{eq:5}.

\begin{theorem}
\label{thm:1}
If assumptions (A1)--(A6) hold and $\lambda \to 0$ in such a way that $n \lambda^{1/m} \to \infty$ and $n \lambda^{1/(2m)} \exp[-\lambda^{-1/{m}}] \to 0$ as $n \to \infty$. Then,
\begin{align*}
\lim_{D \to \infty} \liminf_{n \to \infty} \Pr[& \text{there exists a sequence of local minimizers $\widehat{g}_n$ of \eqref{eq:5} satisfying} \\  &\ \left\|\widehat{g}_n - g_0\right\|_{m, \lambda}^2 \leq D(n^{-1} \lambda^{-1/(2m)} + \lambda)   ] = 1.
\end{align*}
\end{theorem}
\noindent
The result in Theorem~\ref{thm:1} is local in nature as the objective function \eqref{eq:6} is not convex, unless $\widehat{\alpha}_n = 0$. Similar considerations exist in \citet{Basu:1998}, \citet{Fan:2001} and \citet{Wang:2013}. The limit requirements of the theorem are met, e.g., whenever $\lambda \asymp n^{-2m/(2m+1)}$, in which case we are led to the optimal rate $\|\widehat{g}_n - g_0\|_{m, \lambda}^2 = O_P(n^{-2m/(2m+1)})$. This rate is similar to the rate obtained in nonparametric regression for continuous responses \citep{Kalogridis:2021}. A faster $n^{-4m/(4m+1)}$-rate can be obtained whenever $g_0$ and its derivatives fulfil certain boundary conditions, see \citet{Eg:2009, Kalogridis:2021}. Since we assume that, for all large $n$, $\widehat{\alpha}_n>0$ with high probability, the theorem does not cover penalized likelihood estimators, but these estimators may be covered with similar arguments as in \citep{Mammen:1997, van de Geer:2000} and the same rates of convergence would be obtained.

As for all $\lambda>0$, $\| \cdot \| < \|\cdot\|_{m, \lambda}$ the same rate applies for the more commonly used $\mathcal{L}^2([0,1])$-norm. However, the fact that our results are stated in terms of the stronger $\|\cdot\|_{m,\lambda}$ leads to two notable consequences. First, for $\lambda \asymp n^{-2m/(2m+1)}$ the bound in \eqref{eq:7} immediately yields
\begin{align*}
\sup_{t \in [0,1]}|\widehat{g}_n(t) - g_0(t)| = O_P(n^{(1/2-m)/(2m+1)}),
\end{align*}
which implies that convergence can be made uniform. Although this uniform rate is slower than the $\log^{1/2}(n) n^{-m/(2m+1)}$-rate in \citet[Chapter 21]{Eg:2009} for the classical smoothing spline in case of continuous data with constant variance, this rate is guaranteed for a much broader setting encompassing many response distributions. Secondly, by using Sobolev embeddings we can also describe the rate of convergence of the derivatives of $\widehat{g}_n$ in terms of the $\mathcal{L}^2([0,1])$-metric. These are given in Corollary~\ref{cor:1} below. To the best of our knowledge, these are the first results on uniform convergence and convergence of derivatives in among all robust estimation methods for nonparametric generalized linear models.
\begin{cor}
\label{cor:1}
Assume the conditions of Theorem~\ref{thm:1} hold. Then, for any $\lambda \asymp n^{-2m/(2m+1)}$, the sequence of minimizers in Theorem~\ref{thm:1} satisfies
\begin{align*}
\left\|\widehat{g}_n^{(j)}-g_0^{(j)}\right\|^2 = O_P(n^{-2(m-j)/(2m+1)}), \quad (j = 1, \ldots, m).
\end{align*}
\end{cor}

\subsection{Penalized spline type estimators}

The assumptions for the penalized spline type estimators are for the most part identical to those for smoothing spline type estimators. In addition to (A1)--(A5), we require the following assumptions.
\begin{itemize}
\item[(B6)] The number of knots $K = K_n \to \infty$ and there exists a $\delta^{\prime}>0$ such that $n^{\delta^{\prime}-1} K^2 \to 0$ as $n \to \infty$.
\item[(B7)] Let $h_i = x_i-x_{i-1}$ and $h = \max_i h_i$. Then, $\max_i |h_{i+1}-h_i| = o(K^{-1})$ and there exists a finite $M>0$ such that $(h/\min_i h_i) \leq M$.
\item[(B8)] Let $Q_n$ denote the empirical distribution of  the design points $t_i, i = 1, \ldots, n$. Then, there exists a distribution function $Q$ with corresponding Lebesgue density bounded away from zero and infinity such that $\sup_{t \in [0,1]}|Q_n(t) - Q(t)| = o(K^{-1})$.
\end{itemize}

Assumptions (B6)--(B8) have been extensively used in the treatment of penalized spline estimators, \citep[see, e.g.,][]{Claeskens:2009, Kau:2009, Xiao:2019}. Assumption (B6) imposes a weak condition on the rate of growth of the interior knots which is not restrictive in practice, as it is in line with the primary motivation behind penalized spline type estimators. Assumption (B7) concerns the placement of the knots and is met if the knots are equidistant, for example. Finally, assumption (B8) is the lower-rank equivalent of assumption (A6) and holds in similar settings.

Recall that $S_{K}^p([0,1])$ is a finite-dimensional space and, for any $f: [0,1] \to \mathbbm{R}$ which is continuous or a step function, $\|f\| = 0$ implies that $f = 0$. Consequently, endowing $S_{K}^p([0,1])$ with $\langle \cdot, \cdot \rangle_{m,\lambda}$ makes it a Hilbert space. Proposition~\ref{prop:1} shows that $S_{K}^p([0,1])$ is an RKHS allowing us to draw direct parallels between $\mathcal{W}^{m,2}([0,1])$ and its dense subspace $S_{K}^p([0,1])$.
\begin{prop}
\label{prop:1}
If assumptions (B6)--(B7) hold and $p>m \geq 1$, then there exists a positive constant $c_0$ such that for every $f \in S_{K}^p([0,1])$, $K \geq 1$ and $\lambda \in [0,1]$ it holds that
\begin{align*}
\sup_{t\in [0,1]} |f(t)| \leq c_0 \min\{K^{1/2}, \lambda^{-1/(4m)}\}\|f\|_{m,\lambda},
\end{align*}
where for $\lambda = 0$ we define $\lambda^{-1/(4m)} = \infty$.
\end{prop}

Proposition~\ref{prop:1} implies the existence of a symmetric function $\mathcal{R}_{m,K,\lambda}: [0,1] \times [0,1] \to \mathbbm{R}$ depending on $K$, $\lambda$, $m$ and $p$ such that, for every $y \in [0,1]$ the map $x \mapsto \mathcal{R}_{m,K,\lambda}(x,y) \in S_{K}^p([0,1])$ and for every  $f \in S_{K}^p([0,1])$, $f(t) = \langle \mathcal{R}_{m,K,\lambda}(t, \cdot), f \rangle_{m,\lambda}$. 
Hence, $\mathcal{R}_{m,K,\lambda}$ is the reproducing kernel which can be derived explicitly in this setting. Let $\mathbf{H}_p$ denote the $(K+p) \times (K+p)$ matrix consisting of the inner products $\langle B_i, B_j \rangle$ with $B_1, \ldots, B_{K+p}$, the B-spline functions of order $p$, and let $\mathbf{P}_{m}$ denote the penalty matrix with elements $\langle B_i^{(m)}, B_j^{(m)} \rangle$.
Set $\mathbf{G}_{\lambda} = \mathbf{H}_p + \lambda \mathbf{P}_m$, then
\begin{align*}
\mathcal{R}_{m,K,\lambda}(x,y) = \mathbf{B}^{\top}(x) \mathbf{G}_{\lambda}^{-1} \mathbf{B}(y).
\end{align*}
Since, for $f, g \in S_{K}^p([0,1])$ we have $\langle f, g \rangle_{m,\lambda} = \mathbf{f}^{\top} \mathbf{G}_{\lambda} \mathbf{g}$, with $\mathbf{f},\mathbf{g} \in \mathbbm{R}^{K+p}$ the vectors of scalar coefficients, it is easy to see that $\mathcal{R}_{m,K,\lambda}$ satisfies the required properties.

Since $S_{K}^p([0,1]) \subset \mathcal{W}^{m,2}([0,1])$ for $p>m$ it follows from \eqref{eq:7} that for all $f \in S_{K}^p([0,1])$ it holds that $\sup_{t\in [0,1]} |f(t)| \leq c_0 \lambda^{-1/(4m)}\|f\|_{m,\lambda}$. However, the bound in Proposition~\ref{prop:1} is tighter if $K = o(\lambda^{-1/(2m)})$, that is, if the rate of growth of $K$ is outpaced by the rate of decay of $\lambda^{1/(2m)}$. This suggests that the rate of growth of $K$ and the rate of decay of $\lambda$ jointly determine the asymptotic properties of penalized spline estimators. This relation is formalized in Theorem~\ref{thm:2}.

\begin{theorem}
\label{thm:2} Suppose that assumptions (A1)--(A5) and (B6)--(B8) hold and assume that $\lambda \to 0$ in a such way that $K \min\{\lambda^2 K^{2m}, \lambda \} \to 0$ and $n \max\{K^{-1}, \lambda^{1/(2m} \} \exp[ - \min\{K^2, \lambda^{-1/m}\} ] \to 0 $ as $n \to \infty$. Then, if $g_0 \in \mathcal{C}^{j}([0,1])$ with $m \leq j \leq p$,	
\begin{align*}
\lim_{D \to \infty} \liminf_{n \to \infty} & \Pr[  \text{there exists a sequence of local minimizers $\widehat{g}_n$ of \eqref{eq:6} satisfying} \\   &  \left\|\widehat{g}_n - g_0\right\|^2 \leq D(n^{-1}\min\{K, \lambda^{-1/(2m)}\} + \min\{ \lambda^2 K^{2m}, \lambda \} + K^{-2j})   ] = 1.
\end{align*}
\end{theorem}

Theorem~\ref{thm:2} presents the $\mathcal{L}^2([0,1])$-error as a decomposition of three terms accounting for the variance, the modelling bias and the approximation bias of the estimator, respectively. It is interesting to observe that, apart from the term $K^{-2j}$ stemming from the approximation of a generic $\mathcal{C}^{j}([0,1])$-function with a spline, the error simultaneously depends on $K$ and $\lambda$, highlighting the intricate interplay between knots and penalties in the asymptotics of penalized spline estimators. 

For $K< \lambda^{-1/(2m)}$ Theorem~\ref{thm:2} leads to the regression spline asymptotics established by \citet{Claeskens:2009, Xiao:2019} for Gaussian responses (see also \citet{Kim:2004}) whereas for  $g_0 \in \mathcal{C}^m([0,1])$ and $K \geq \lambda^{-1/(2m)}$ we obtain
\begin{align*}
\left\|\widehat{g}_n - g_0\right\|^2 = O_P(n^{-1}\lambda^{-1/(2m)}) + O_P(\lambda)  + O_P(K^{-2m}),
\end{align*}
which, apart from the approximation error $K^{-2m}$, corresponds to the $\mathcal{L}^2([0,1])$ error decomposition of smoothing spline estimators in Theorem~\ref{thm:1}. This fact has important practical implications, as it allows for an effective dimension reduction without any theoretical side effects. Indeed, taking $\lambda \asymp n^{-2m/(2m+1)}$ and $K \asymp n^{\gamma}$ for any $\gamma \geq 1/(2m+1)$ leads to $\|\widehat{g}_n - g_0\|^2 = O_P(n^{-2m/(2m+1)})$, which is the same rate of convergence as for smoothing spline estimators. Moreover, with this choice of tuning parameters, the convergence rates of the derivatives given in Corollary~\ref{cor:1} carry over to the present lower-rank estimators, as summarized in Corollary~\ref{cor:2}.
\begin{cor}
\label{cor:2}
Assume the conditions of Theorem~\ref{thm:2} hold and $g_0 \in \mathcal{C}^m([0,1])$. Then, for any $\lambda \asymp n^{-2m/(2m+1)}$ and $K \asymp n^{\gamma}$ with $\gamma \geq 1/(2m+1)$,
\begin{align*}
\left\|\widehat{g}_n^{(j)}-g_0^{(j)}\right\|^2 = O_P(n^{-2(m-j)/(2m+1)}), \quad (j = 1, \ldots, m).
\end{align*}
\end{cor}
While there have been attempts in the past to cast penalized spline estimators in an RKHS framework \citep[e.g., in][]{Pearce:2006}, this was from a computational perspective. To the best of our knowledge, our theoretical treatment of penalized spline estimators based on the theory of RKHS is the first of its kind in both the classical and the robustness literature and may be of independent mathematical interest. The interested reader is referred to the accompanying supplementary material for the technical details.

\section{Practical implementation}
\label{sec:4}

\subsection{Computational algorithm}

The practical implementation of the smoothing and penalized spline estimators based on density power divergence requires a computational algorithm as well as specification of their parameters, namely the tuning parameter $\alpha$, the penalty parameter $\lambda$ and $K$, the  dimension of the spline basis, in case of penalized splines. Practical experience with penalized spline estimators has shown that the dimension of the spline basis is less important than the penalty parameter, provided that $K$ is taken large enough \citep{Ruppert:2003, Wood:2017}. Hence, little is lost by selecting $K$ in a semi-automatic manner. We now discuss a unified computational algorithm for the smoothing and penalized type estimators described in Section~\ref{sec:3} and briefly discuss the selection of $K$.

By using the B-spline basis, the computation of the proposed estimators can be unified. Recall from our discussion in Section~\ref{sec:3} that a solution to \eqref{eq:5} is a natural spline of order $2m$. Assume for simplicity that all the $t_i$ are distinct and that $\min_{i} t_i>0$ and $\max_{i} t_i<1$. Then, the natural spline has $n$ interior knots and we may represent the candidate minimizer $g \in S_{K}^p([0,1])$ as $g = \sum_{k = 1}^{n + 2m} g_k B_{k}$ where the $B_k$ are the B-spline basis functions of order $2m$ supported by the knots at the interior points $t_i$ and the $g_k$ are scalar coefficients \citet[p. 102]{DB:2001}. The penalty term now imposes the boundary conditions. The reasoning is as follows: if that were not the case, it would always be possible to find a $2m$th order natural spline $\widehat{h}_n(t)$, written in terms of $B_1, \ldots, B_{K+p}$, that interpolates $\widehat{g}_n(t_i), i = 1, \ldots, n$ leaving the first term in \eqref{eq:5} unchanged, but due to it being a polynomial of order $m$ outside $[\min_{i} t_i, \max_{i} t_i] \subset [0,1]$ we would always have $\|\widehat{h}_n^{(m)}\| <  \|\widehat{g}_n^{(m)}\|$. 

The above discussion shows that to compute either the smoothing spline type or the penalized spline type estimators it suffices to minimize
\begin{align}
\label{eq:8}
L_n(g) =   \sum_{i=1}^n l_{\widehat{\alpha}_n}(Y_i,g(t_i)) + \lambda \int_{[0,1]} |g^{(m)}(t)|^2 dt,
\end{align}
over $g \in S_{K}^p([0,1])$, where in the smoothing spline case the knots satisfy $x_i = t_i, i =1, \ldots, n$ and the order $p = 2m$ while in the penalized spline case $0< x_1 < \ldots < x_K <1$ and $p > m$. By default we set $p = 2m$ in our implementation of the algorithm and use all interior knots, i.e., $x_i = t_i$ for $i \leq n$ if $n \leq 50$. For $n>50$ we adopt a thinning strategy motivated by the theoretical results of the previous section as well as the strategy employed by the popular \textsf{smooth.spline} function in the \textsf{R} language \citep[see][p. 189]{Hastie:2009}. In particular, for $n>50$ we only employ $K \simeq n^{1/(2m+1)}$ interior knots, which we spread in the $[0,1]$-interval in an equidistant manner. While this leads to large computational gains, Theorem~\ref{thm:2} assures that no accuracy is sacrificed. For example, with $m=2$ and $n = 5000$ our strategy amounts to using only 83 knots; a dramatic dimension reduction.

We solve \eqref{eq:8} with a Newton-Raphson procedure, which we initiate from the robust estimates of \citet{Kal:2021} and \citet{Croux:2012}. The updating steps of the algorithm can be reduced to penalized iteratively reweighted least-squares updates, in the manner outlined by \citet[Chapter 5.2.3]{Green:1994}. The weights in our case are given by
\begin{align*}
w_{i} &= (1+\widehat{\alpha}_n)^2 \mathbb{E}_{g_i} \{ |Y_i - b^{\prime}(g_i)|^2 f_{g_i}^{\widehat{\alpha}_n}(Y_i)\} - (1+\alpha) b^{\prime \prime}(g_i) \mathbb{E}_{g_i}\{f_{g_i}^{\widehat{\alpha}_n}(Y_i) \}
\\ & \quad + (1+\widehat{\alpha}_n) b^{\prime \prime}(g_i) f_{g_i}^{\widehat{\alpha}_n}(Y_i) - \widehat{\alpha}_n(1+\widehat{\alpha}_n)|Y_i - b^{\prime}(g_i)|^2 f_{g_i}^{\widehat{\alpha}_n}(Y_i), \quad (i=1, \ldots, n),
\end{align*}
where $g_i = \mathbf{B}^{\top}(t_i) \boldsymbol{g}$ and the vector of ``working" data $\mathbf{z} \in \mathbbm{R}^n$ consists of
\begin{align*}
z_i = g_i - (1+\widehat{\alpha}_n)\frac{ \mathbb{E}_{g_i}\{(Y_i-b^{\prime}(g_i))f_{g_i}^{\widehat{\alpha}_n}(Y_i) \} - \{Y_i-b^{\prime}(g_i)\} f_{g_i}^{\widehat{\alpha}_n}(Y_i)}{w_{i}}.\end{align*}
Alternatively, taking expected values we may replace the weights with $(1+\widehat{\alpha}_n) \mathbb{E}_{g_i}\{|Y_i-b^{\prime}(g_i)|^2f_{g_i}^{\widehat{\alpha}_n}(Y_i)\}$ thereby obtaining a variant of the Fisher scoring algorithm. For $\widehat{\alpha}_n = 0$, these formulae reduce to those for penalized likelihood estimators and canonical links, cf. \citet{Green:1994}.

\subsection{Selection of the tuning and penalty parameters}

The tuning parameter $\alpha$ determines the trade-off between robustness and efficiency of the estimators. Selecting $\alpha$ independent of the data could lead to lack of resistance towards atypical observation or an undesirable loss of efficiency. To determine $\alpha$ in a data-driven way we modify the strategy of \citet{Ghosh:2015} and rely on a suitable approximation of the integrated mean integrated squared error (MISE) of $\widehat{g}_n$, i.e.,  $\mathbb{E}\{\left\|\widehat{g}_{n}- g_0\right\|^2 \}$. To show the dependence of the estimator on $\alpha$ we now denote it by $\widehat{g}_{n,\alpha}$. For each $\alpha \geq 0$ its MISE can be decomposed as 
\begin{align}
\label{eq:9}
\MISE(\alpha) = \left\| \mathbb{E}\{\widehat{g}_{n,\alpha} \} - g_0\right\|^2 + \mathbb{E}\{\left\|\widehat{g}_{n,\alpha} - \mathbb{E}\{\widehat{g}_{n,\alpha}\}\right\|^2\},
\end{align}
where the first term on the RHS represents the integrated squared bias, while the second term is the integrated variance of the estimator. Neither of these terms can be computed explicitly since $\mathbb{E}\{\widehat{g}_{n,\alpha} \}$ and $g_0$ are both unknown. Therefore, we seek approximations. 
Following \citet{War:2005} and \citet{Ghosh:2015}, we replace $\mathbb{E}\{\widehat{g}_{n,\alpha} \}$ in the bias term by $\widehat{g}_{n,\alpha}$ and use a ``pilot" estimator instead of the unknown $g_0$. To limit the influence of aberrant observations on the selection of $\alpha$ we propose to replace $g_0$ by the robust estimate $\widehat{g}_{n,1}$ which minimizes the $L_2$-distance between the densities, as described in Section~\ref{sec:2}.

To approximate the variance term, observe that $\mathbb{E} \{\left\|\widehat{g}_{n,\alpha} - \mathbb{E}\{\widehat{g}_{n,\alpha}\}\right\|^2\} = \Tr\{ \mathbf{H}_p \Cov\{\widehat{\mathbf{g}}_{n,\alpha} \} \}$.
Using the notation  
\begin{align*}
l^{\prime}_{\alpha}(Y_i,g(t_i)) = \frac{\partial l_{\alpha}(Y_i,x)}{\partial x} \bigg \vert_{x = g(t_i)}, \quad (i=1, \ldots, n).
\end{align*}
for the first derivatives and analogously for the second derivatives $l^{\prime \prime}_{\alpha}(Y_i,g(t_i))$, define $\mathbf{C}_{\alpha} = \diag\{|l^{\prime}_{\alpha}(Y_1,\widehat{g}_{n,\alpha}(t_1))|^2, \ldots, |l^{\prime}_{\alpha}(Y_n, \widehat{g}_{n,\alpha}(t_n))|^2 \}$ as well as $\mathbf{D}_{\alpha}  = \diag\{l^{\prime \prime }_{\alpha}(Y_1,\widehat{g}_{n,\alpha}(t_1)), \ldots, l^{\prime \prime}_{\alpha}(Y_n, \widehat{g}_{n,\alpha}(t_n)) \}$. A first order Taylor expansion of the score function of \eqref{eq:8} readily yields
\begin{align*}
\Cov\{\widehat{\mathbf{g}}_{n,\alpha}\} \approx \left[ \mathbf{B}^{\top} \mathbf{D}_{\alpha} \mathbf{B} + 2 \lambda \mathbf{P}_m \right]^{-1} \mathbf{B}^{\top} \mathbf{C}_{\alpha} \mathbf{B} \left[ \mathbf{B}^{\top} \mathbf{D}_{\alpha} \mathbf{B} + 2 \lambda \mathbf{P}_m \right]^{-1},
\end{align*}
where $\mathbf{B}$ is the $n \times (K+p)$ spline design matrix with $ij$th element given by $B_j(t_i)$.

Inserting the approximations of the bias and variance in \eqref{eq:9} we obtain the approximate mean-squared error (AMISE) given by
\begin{align*}
\AMISE(\alpha) & =  (\widehat{\mathbf{g}}_{n,\alpha} - \widehat{\mathbf{g}}_{n,1})^{\top} \mathbf{H}_p (\widehat{\mathbf{g}}_{n,\alpha} - \widehat{\mathbf{g}}_{n,1}) \\  & \ +  \Tr\{ \mathbf{H}_p \left[ \mathbf{B}^{\top} \mathbf{D}_{\alpha} \mathbf{B} + 2 \lambda \mathbf{P}_m \right]^{-1} \mathbf{B}^{\top} \mathbf{C}_{\alpha} \mathbf{B} \left[ \mathbf{B}^{\top} \mathbf{D}_{\alpha} \mathbf{B} + 2 \lambda \mathbf{P}_m \right]^{-1} \}.
\end{align*}
We propose to select $\widehat{\alpha}_n$ by minimizing $\AMISE(\alpha)$ over a grid of $20$ equidistant candidate values in $[0,1]$. This grid includes both the maximum likelihood ($\alpha = 0$) and $L_2$-distance estimators ($\alpha  = 1$) as special cases.

The bias approximation and thus the selection of $\alpha$ depends on the pilot estimator  $\widehat{g}_{n,1}$. To reduce this dependence \citet{Basak:2021} propose to iterate the selection procedure. That is, using $\widehat{g}_{n,1}$ as pilot estimator determine the value $\widehat{\alpha}$ minimizing $\AMISE(\alpha)$ and use the corresponding density power divergence estimate as the new ``pilot" estimate instead of $\widehat{g}_{n,1}$. This procedure is repeated until $\widehat{\alpha}$ converges, which in our experience takes between $1$ and $3$ iterations for the vast majority of cases in our setting. 

The computation of the estimator for a given value of $\alpha$ requires an appropriate value of the penalty parameter $\lambda$. To determine $\lambda$, we utilize the Akaike Information Criterion in the form proposed by \citet{Hastie:1990} given by
\begin{align*}
\AIC(\lambda) = 2 \sum_{i=1}^n l_{\alpha}(Y_i,\widehat{g}_{n, \alpha}(t_i)) + 2 \Tr\{ \left[ \mathbf{B}^{\top} \mathbf{D}_{\alpha} \mathbf{B} + 2 \lambda \mathbf{P}_m \right]^{-1} \mathbf{B}^{\top} \mathbf{D}_{\alpha} \mathbf{B} \}.
\end{align*}
Implementations and illustrative examples of the density power divergence smoothing/penalized type spline estimators are available at \url{https://github.com/ioanniskalogridis/Robust-and-efficient-estimation-of-nonparametric-GLMs}.

\section{Finite-sample performance}
\label{sec:5}

We now illustrate the practical performance of the proposed density power divergence spline estimators for GLM with Gaussian, Binomial and Poisson responses. We compare the following estimators.
\begin{itemize}
\item The adaptive density power divergence estimator discussed in the previous Section, denoted by DPD($\widehat{\alpha}$).
\item The $L_2$-distance estimator corresponding to $\alpha=1$, denoted by DPD($1$).
\item The standard penalized maximum likelihood estimator corresponding to $\alpha=0$, abbreviated  as GAM.
\item The robust Huber-type P-spline estimator of \citet{Croux:2012} with 40 B-spline functions, denoted by RQL, as implemented in the \textsf{R}-package \textsf{DoubleRobGam}.
\item The robust local polynomial estimator of \citet{Ali:2011}, denoted by RGAM, as implemented in the \textsf{R}-package \textsf{rgam}.
\end{itemize}
For the first three estimators we use the default settings described in Section~\ref{sec:4} with B-splines of order $p=4$ combined with penalties of order $m=2$. For Gaussian data, we use the resistant Rice-type estimator discussed in \citet{Kalogridis:2021} to estimate $\phi$. This estimate is used for the DPD$(\widehat{\alpha})$, DPD$(1)$ and RQL estimators. Since RGAM is not implemented for Gaussian responses, we have replaced it in this case by the robust local linear estimator obtained with the \textsf{loess} function and reweighting based on Tukey's bisquare, see \citet{Clev:1979}. Both RQL and RGAM are by default tuned for nominal 95\% efficiency for Gaussian data. The penalty parameter is selected via robust AIC for RQL and robust cross-validation for RGAM.

Our numerical experiments assess the performance of the competing estimators on both ideal data and data that contain a small number of atypical observations. For our experiments we set  $t_i = i/(n+1), i = 1, \ldots, n,$ and consider the following two test functions:
\begin{itemize}
\item $g_1(t) = -\sin(25t/6)/0.8 - 1$,
\item $g_2(t) = 1.8 \sin(3.4x^2)$.
\end{itemize}
which were also considered by \citet{Ali:2011} and \citet{Croux:2012}. For Gaussian responses we generate each $Y_i$ from a mixture of two Gaussian distributions with mixing parameter $\epsilon$, 
equal mean $g_j(t_i)$, $j \in \{1, 2\}$ and standard deviations equal to $1$ and $9$, respectively. This yields a proportion $(1-\epsilon)$ of ideal data distorted with a proportion $\epsilon$ of outliers originating from a Gaussian distribution with heavier tails.

For Binomial and Poisson responses each $Y_i$ has mean $\mu_j(t_i) = G^{-1}(g_j(t_i))$ for $j=1$ or $2$ with $G$ the respective canonical links. We then contaminate a fraction $\epsilon$ of the responses as follows. For Binomial data $Y_i$ is set to $0$ if the original value was equal to $1$ and vice versa. For Poisson data $Y_i$ is replaced by a draw from a Poisson distribution with mean (and variance) equal to $3\mu_j(t_i)$. In the Binomial case, this contamination represents the frequently occurring situation of missclassified observations while we have  a small number of overdispersed observations in the Poisson case. The contamination fraction $\epsilon$ takes values in $\{0, 0.05, 0.1\}$ reflecting increasingly severe contamination scenarios. For each setting we generated $1000$ samples of size $n = 200$. We evaluate the performance of the competing estimators via their mean-squared error (MSE), given by $\sum_{i=1}^n |\widehat{\mu}_i - \mu_i|^2/n$. 
Since, MSE distributions are right-skewed in general, we report both the average and median of the MSEs.
Tables 1--3 report the results for data with Gaussian, Binomial and Poisson responses, respectively.
\begin{table}[h]
\centering
\caption{Mean and median MSEs ($\times 100$) for Gaussian responses from 1000 simulated datasets with $n = 200$.}
\label{tab:1}
\resizebox{\textwidth}{!}{%
\begin{tabular}{c c cc cc cc cc cc cc}
& & \multicolumn{2}{c}{DPD($\widehat{\alpha}$)} & \multicolumn{2}{c}{DPD(1)} &   \multicolumn{2}{c}{GAM} & \multicolumn{2}{c}{RQL} &  \multicolumn{2}{c}{RGAM} \\ \\[-2ex]
$g_0$ & $\epsilon$ & Mean & Median & Mean & Median & Mean & Median & Mean & Median & Mean & Median
\\ \\ [-1.2ex]
\multirow{3}{*}{$g_1$} & 0  & 2.98 & 2.28 & 5.29  & 4.29  & 2.81 & 2.22 &  3.23 & 2.72 & 2.93 & 2.31  \\ \\[-2ex]
&  0.05 & 3.56 & 2.84 & 5.95  & 4.56  & 12.35 & 8.87 & 4.68 & 4.03 & 12.65 & 9.24  \\ \\[-2ex]
&  0.1 &  4.48 & 3.49 & 5.93 & 4.98 &  20.99 & 15.72 & 7.85 & 6.54 & 22.89 & 16.90
\\ \\[-1.2ex]
\multirow{3}{*}{$g_2$} & $0$ & 4.34 & 3.37 & 7.57 & 6.72 & 3.40 & 2.66 & 4.64 & 3.91 & 3.38 & 2.75  \\ \\[-2ex]
&  $0.05$ & 5.21 & 4.23 & 8.27 & 7.37 & 15.86 & 11.16 & 6.72 & 5.93 & 14.80 & 11.25 \\ \\[-2ex]
&  $0.1$ & 6.24 & 5.26 & 8.67 & 7.88 & 28.52 & 19.95 & 11.19 & 10.21 & 25.85 & 19.11
\end{tabular}}
\end{table}

The simulation experiments lead to several interesting findings. For uncontaminated Gaussian data, Table~\ref{tab:1} shows that 
GAM and RGAM perform slightly better than DPD$(\widehat{\alpha})$ and RQL while the $L_2$-distance estimator DPD$(1)$ is less efficient. However, even with a small amount of contamination, the advantage of GAM and RGAM evaporates. RQL offers more protection against outlying observations, but it is outperformed by DPD($\widehat{\alpha})$ and also by DPD(1) in case of more severe contamination. This difference may be explained by the use of the monotone Huber function in RQL versus the redescending squared exponential loss function in DPD \citep{Maronna:2019}.

As seen in Table~\ref{tab:2}, for binomial data DPD($\widehat{\alpha})$ remains highly efficient and robust. In ideal data, DPD$(\widehat{\alpha})$ even attains a lower mean and median MSE than GAM, although this is likely an artefact of sampling variability. A surprising finding here is the exceptionally good performance of DPD$(1)$ in clean data. It outperforms both RQL and RGAM, performing almost as well as DPD($\widehat{\alpha})$. This fact suggests that efficiency loss for Binomial responses is much less steep as a function of $\alpha$ than for the Gaussian case. See \citet{Ghosh:2016} for efficiency considerations in the parametric case.

\begin{table}[h]
\caption{Mean and median MSEs ($\times 100$) for Binomial responses from 1000 simulated datasets with $n = 200$.}
\label{tab:2}
\centering
\resizebox{\textwidth}{!}{%
\begin{tabular}{c c cc cc cc cc cc cc}
& & \multicolumn{2}{c}{DPD($\widehat{\alpha}$)} & \multicolumn{2}{c}{DPD(1)} &   \multicolumn{2}{c}{GAM} & \multicolumn{2}{c}{RQL} &  \multicolumn{2}{c}{RGAM} \\ \\[-2ex]
$g_0$ & $\epsilon$ & Mean & Median & Mean & Median & Mean & Median & Mean & Median & Mean & Median
\\ \\ [-1.2ex]
\multirow{3}{*}{$g_1$} & 0 & 0.34 & 0.32 & 0.40 & 0.34 & 0.54  & 0.33  & 0.50 & 0.47 & 0.68 & 0.43 \\ \\[-2ex]
&  0.05 & 0.53 & 0.46 & 0.54 & 0.48 & 0.64  & 0.52 & 0.62 & 0.55 & 0.87 &  0.54   \\ \\[-2ex]
&  0.1 & 0.88 & 0.80 & 0.87 & 0.80 & 1.05 & 0.86   & 0.95 & 0.86 & 1.36 & 1.01
\\ \\[-1.2ex]
\multirow{3}{*}{$g_2$} & $0$ & 0.55 & 0.47 & 0.58 & 0.49 & 0.78 & 0.55 & 0.76 & 0.72 & 0.98 & 0.66  \\ \\[-2ex]
&  $0.05$ & 0.65 & 0.55 & 0.66  &  0.58 & 0.84 & 0.63 & 0.87 & 0.83 & 1.18 & 0.75 \\ \\[-2ex]
&  $0.1$ & 0.90 & 0.82 & 0.91 & 0.82 & 1.10 & 0.86 & 1.05 & 1.02 & 1.38 & 0.98
\end{tabular}}
\end{table}

As evidenced from Table~\ref{tab:3}, the situation for count responses is reminiscent of the situation with Gaussian responses. In clean data, GAM exhibits the best median performance, closely followed by DPD($\widehat{\alpha})$. RQL performs slightly worse than DPD($\widehat{\alpha})$ in clean data, but the difference between these two estimators becomes more pronounced in favour of DPD($\widehat{\alpha})$ the heavier the contamination. In ideal data, DPD(1) performs adequately for $\mu_1(t)$, but is rather inefficient for $\mu_2(t)$. 
RGAM performs better for $\mu_2(t)$ than for $\mu_1(t)$, but is outperformed by DPD($\widehat{\alpha})$ in both cases.
Some further insights on the good performance of DPD($\widehat{\alpha})$ are given in the supplementary material.
\begin{table}[h]
\caption{Mean and median MSEs ($\times 100$) for Poisson responses from 1000 simulated datasets with $n = 200$.}
\label{tab:3}
\centering
\resizebox{\textwidth}{!}{%
\begin{tabular}{c c cc cc cc cc cc cc}
& & \multicolumn{2}{c}{DPD($\widehat{\alpha}$)} & \multicolumn{2}{c}{DPD(1)} &   \multicolumn{2}{c}{GAM} & \multicolumn{2}{c}{RQL} &  \multicolumn{2}{c}{RGAM} \\ \\[-2ex]
$g_0$ & $\epsilon$ & Mean & Median & Mean & Median & Mean & Median & Mean & Median & Mean & Median
\\ \\ [-1.2ex]
\multirow{3}{*}{$g_1$} & 0 & 1.02 & 0.80 & 1.03 & 0.82 & 1.15 & 0.76  & 1.17 & 0.94 & 2.57 & 1.61 \\ \\[-2ex]
&  0.05 & 1.13 & 0.94 & 1.21 & 1.04 & 1.67 &  0.98 & 1.27 &  1.00 & 3.15 & 1.93   \\ \\[-2ex]
&  0.1 & 1.39 & 1.06 & 1.40 & 1.10 & 2.76 & 1.75 & 1.72 & 1.34 & 4.27 & 2.69
\\ \\[-1.2ex]
\multirow{3}{*}{$g_2$} & $0$ & 10.61 & 7.62 & 17.70 & 14.57 & 10.12 & 7.62 & 14.93 & 12.05 & 11.89 & 9.24  \\ \\[-2ex]
&  $0.05$ & 13.53 & 10.22 & 18.61 & 15.82 & 59.32 & 44.28 & 15.81 & 13.97 & 15.34 & 12.01  \\ \\[-2ex]
&  $0.1$ & 17.72 & 14.17 & 20.67 &17.26 & 152.7 & 140.1 & 22.57 & 21.34 & 26.06 & 21.20
\end{tabular}}
\end{table}

\section{Applications}
\label{sec:6}

\subsection{Type 2 diabetes in the U.S.}

The National Health and Nutrition Examination Survey is a series of home interviews and physical examinations of randomly selected individuals conducted by the U.S. Centers for Disease Control and Prevention (CDC), normally on a yearly basis. Due to the Covid-19 pandemic, the cycle of data collection between 2019 and 2020 was not completed and, in order to obtain a nationally representative sample, the CDC combined these measurements with data collected in the 2017--2018 period leading to a dataset consisting of 9737 individuals. We refer to \url{https://wwwn.cdc.gov/nchs/nhanes} and the accompanying documentation for a detailed description of the survey and the corresponding data.

For this example, we study of the relationship between type 2 diabetes and high-density-lipoprotein (HDL) cholesterol as well as the relationship between type 2 diabetes and the amount of glycohemoglobin in one's bloodstream. Measuring the concentration of the latter often constitutes an expedient way of detecting diabetes while low HDL cholesterol values are considered a risk factor. The response variable is binary with value $1$ if the interviewee has been diagnosed with type 2 diabetes and $0$ otherwise, while the covariates are continuous. The response variable is plotted versus either covariate in the left and right panels of Figure~\ref{fig:1}.

\begin{figure}[H]
\centering
\subfloat{\includegraphics[width = 0.495\textwidth]{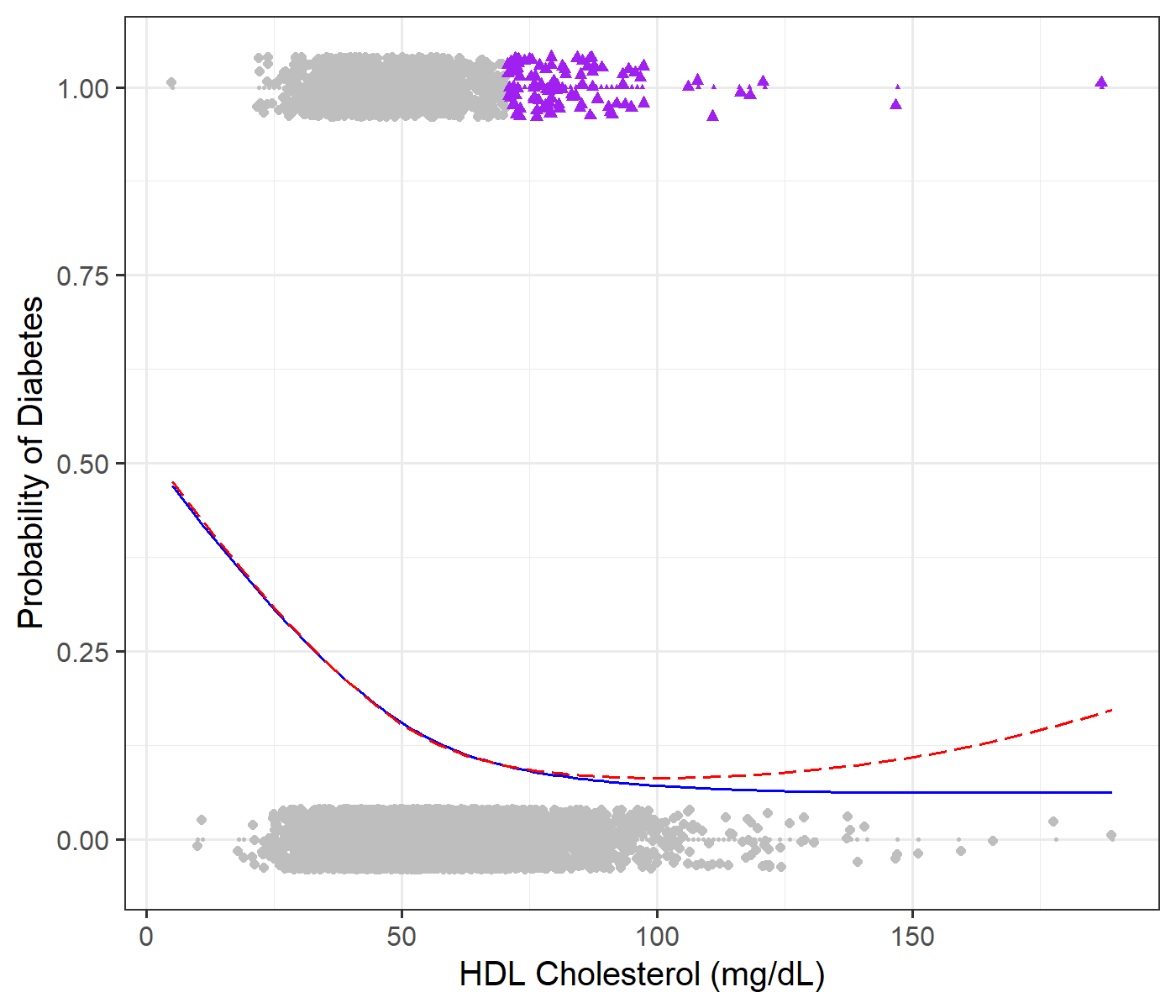}} \
\subfloat{\includegraphics[width = 0.495\textwidth]{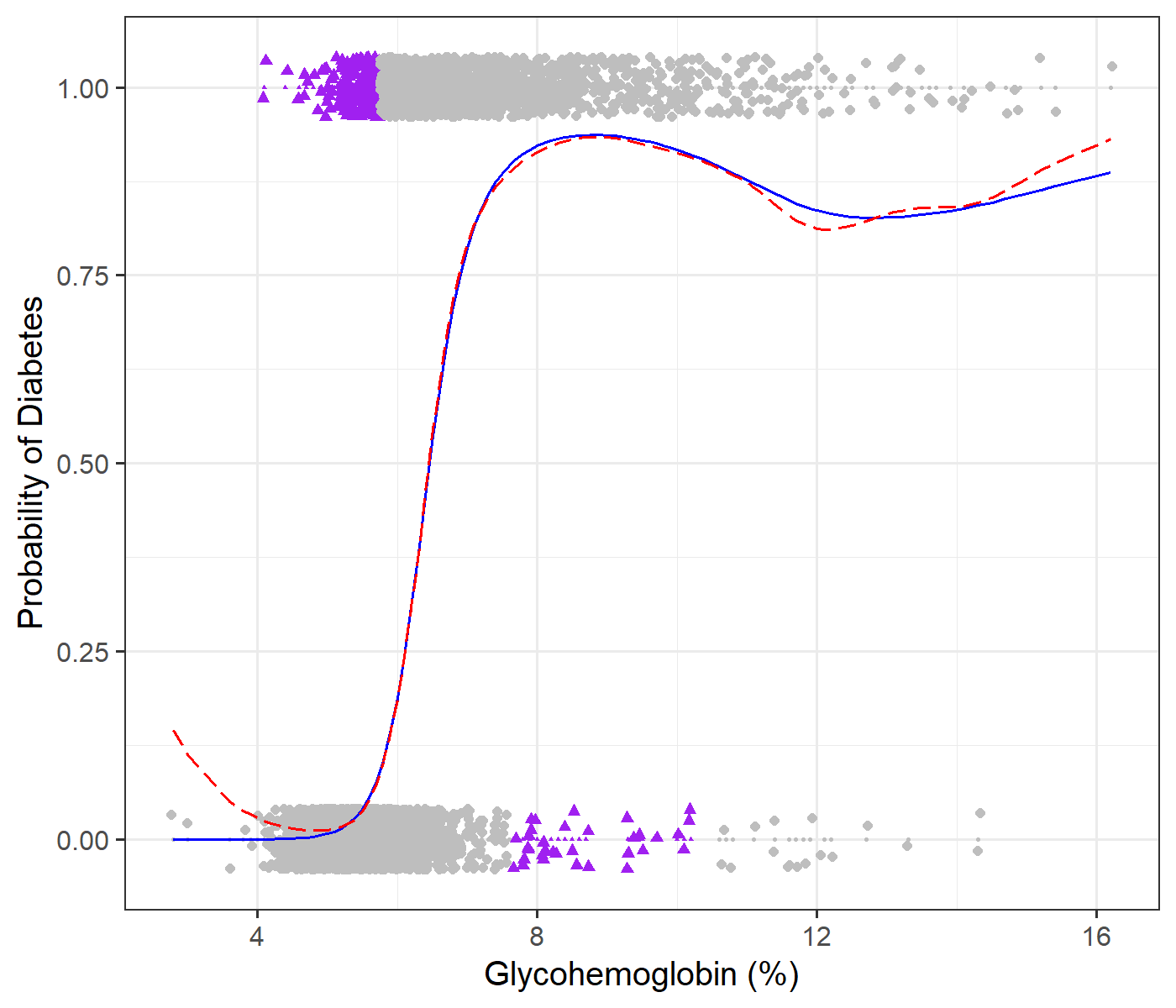}} 		
\caption{Left: Incidence of diabetes versus HDL cholesterol (mg/dL). Right: Incidence of diabetes versus glycohemoglobin (\%). The lines (\full, \denselydashed) correspond to DPD($\widehat{\alpha}$) and GAM estimates respectively. Observations indicated with (\triangled) exhibit large Anscombe residuals according to DPD($\widehat{\alpha}$).}
\label{fig:1}
\end{figure}

Since the way that the covariates influence the probability of diabetes cannot be specified in advance, we apply our methodology to estimate the two regression functions in a nonparametric manner. The algorithm described in Section~\ref{sec:4} selects $\alpha = 1$ in both cases, indicating the presence of several atypical observations among our data. The estimates are depicted with the solid blue lines in the left and right panels of Figure~\ref{fig:1}.  For comparison, the standard GAM estimates obtained with the \texttt{gam} function of the \texttt{mgcv} package \citep{Wood:2017} are depicted by dashed red lines in the figure.

Comparing the robust and GAM-estimates reveals that, despite some areas of agreement, there is notable disagreement between the estimates. In particular, the estimated probabilities of Type 2 diabetes differ significantly for large values of HDL cholesterol and for both small and medium-large concentrations of glycohemoglobin. Inspection of the panels suggests that the GAM estimates are strongly drawn towards a number of atypical observations, corresponding to individuals with high HDL cholesterol but no diabetes and diabetes patients with low and medium concentrations of glycohemoglobin, respectively. The results in both cases are counter-intuitive, as, for healthy individuals with good levels of HDL cholesterol or low levels of glucose, GAM predicts a non-negligible probability of diabetes. By contrast, the robust DPD($\widehat{\alpha}$)-estimates remain unaffected by these atypical observations leading to more intuitive estimates.

Since robust estimates are less attracted to outlying observations, such observations can be detected from their residuals. For GLMs we may make use of Anscombe residuals \citep[p. 29]{MCN:1983}, which more closely follow a Gaussian distribution than their Pearson counterparts. For Bernoulli distributions, these are given by
\begin{align*}
r_{A,i} = \frac{\IB(Y_i, 2/3,2/3) - \IB(\widehat{\mu}_i, 2/3,2/3)}{\widehat{\mu}_i^{1/6}(1-\widehat{\mu}_i)^{1/6}}, \quad (i=1, \ldots, n),
\end{align*}
where $\IB(x,a,b) =\int^x_0 t^{a-1} (1-t)^{b-1} dt$. We classify an observation as an outlier if $|r_{A,i}| \geq 2.6$, which is a conventional cut-off value for the standard Gaussian distribution. The outliers for our examples are shown in Figure~\ref{fig:1} with a different shape and color coding. These plots show that these outliers are largely located in the areas in which the DPD($\widehat{\alpha})$ and GAM-estimates differ, thus confirming the sensitivity of GAM-estimates.

\begin{figure}
\centering
\subfloat{\includegraphics[width = 0.495\textwidth]{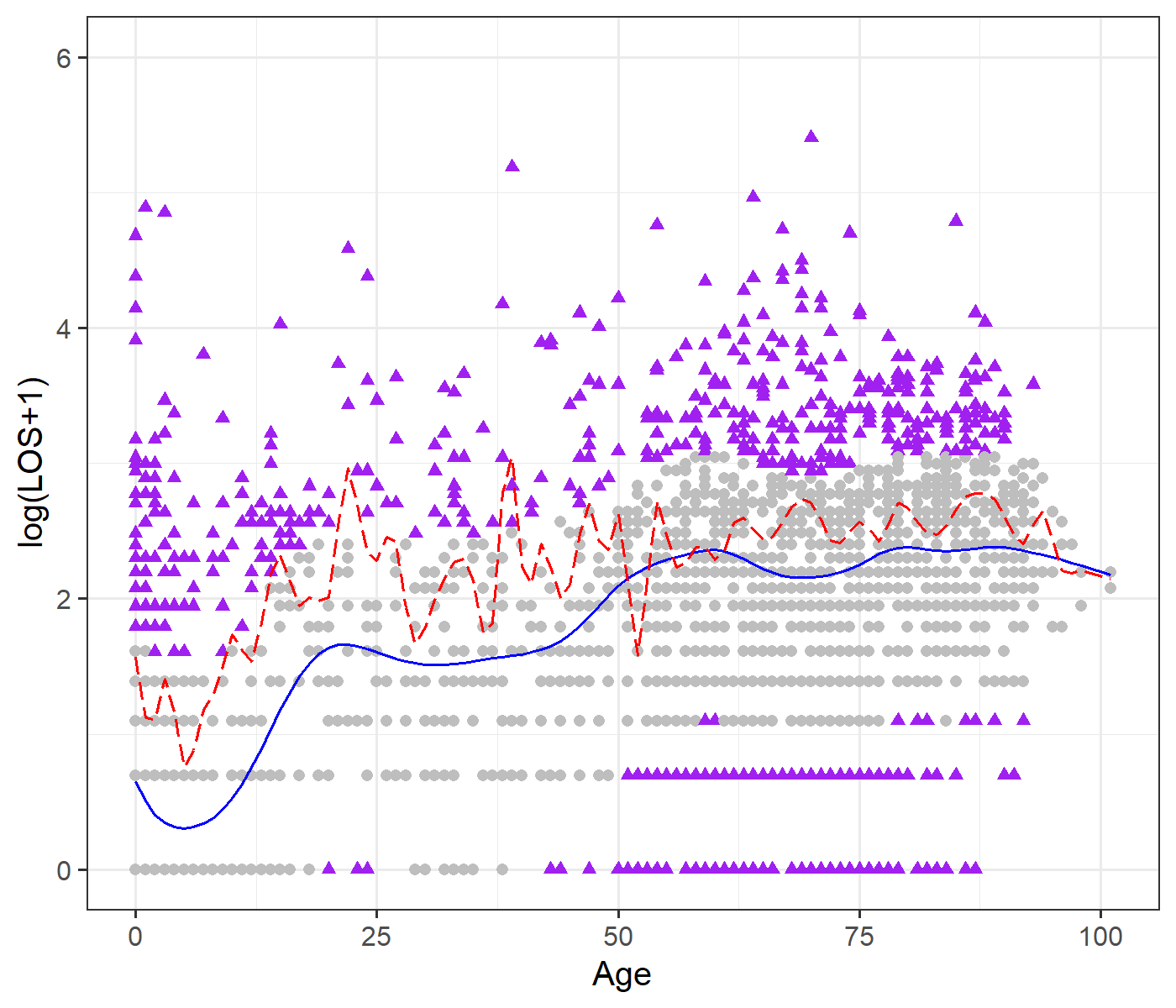}} \
\subfloat{\includegraphics[width = 0.495\textwidth]{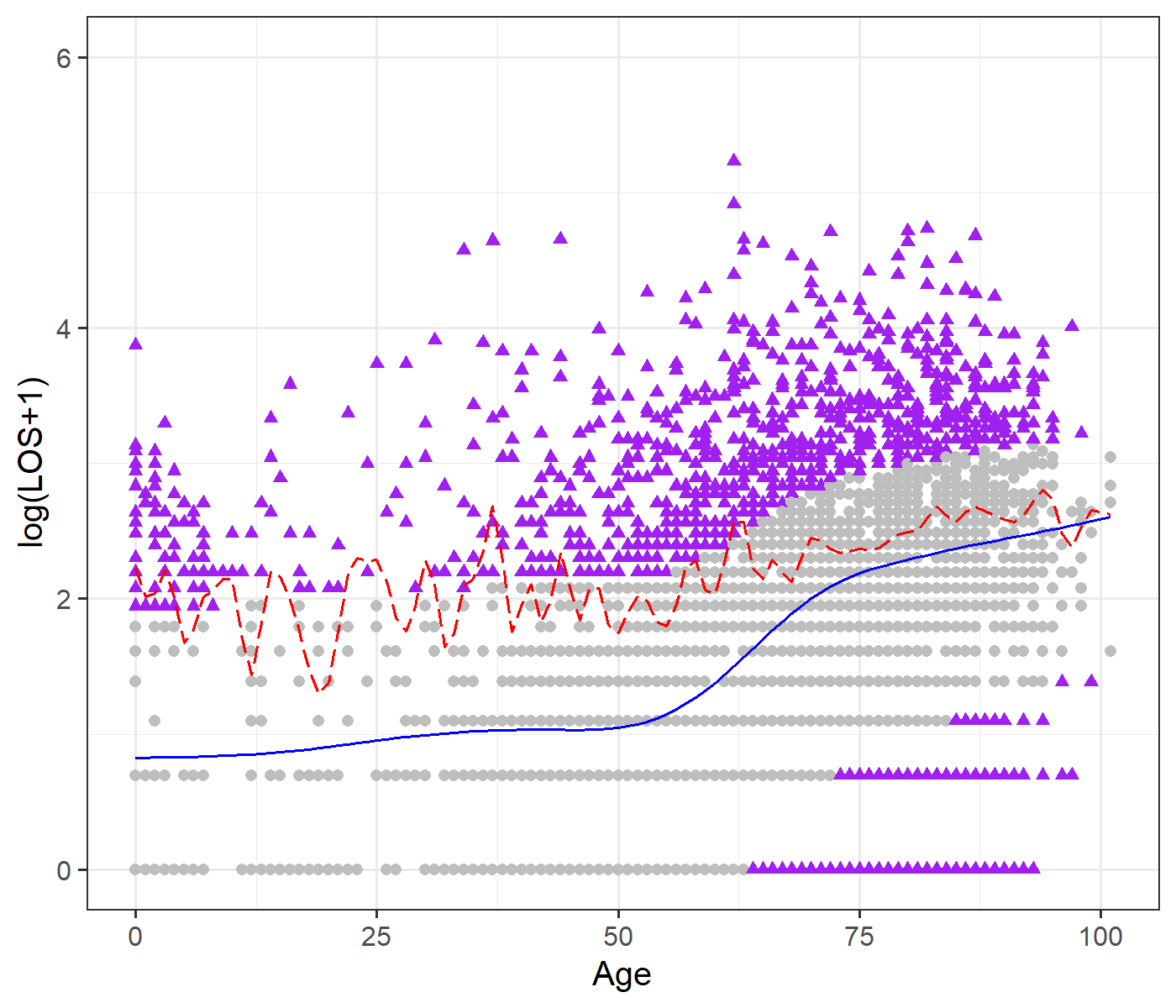}}	
\caption{Left: Logarithm of LOS$+1$ versus age for patients with disorders of the respiratory system. Right: Logarithm of LOS$+1$ versus age for patients with disorders of the circulatory system. The lines (\full, \denselydashed) correspond to DPD($\widehat{\alpha}$) and GAM estimates respectively. Observations indicated with (\triangled) exhibit large Anscombe residuals according to DPD($\widehat{\alpha}$).}	
\label{fig:2}
\end{figure}
\subsection{Length of hospital stay in Switzerland}

Patients admitted into Swiss hospitals are classified into homogeneous diagnosis related groups (DRG). The length of stay (LOS) is an important variable in the determination of the cost of treatment so that the ability to predict LOS from the characteristics of the patient is helpful. Herein, we use the age of the patient (in years) to predict the average LOS (in days) for two DRG comprising \textit{Diseases and disorders of the respiratory system} and \textit{Diseases and disorders of the circulatory system}, which we henceforth abbreviate as DRG 1 and DRG 2. We assume that LOS can be modelled as a Poisson random variable with
\begin{align*}
\log( \mathbb{E}\{\text{LOS}_j \}) = g_j(\text{Age}), \quad j \in \{1, 2\},
\end{align*}
for unknown functions $g_1$ and $g_2$ corresponding to DRG 1 and DRG 2, respectively. The data consisting of 2807 and 3922 observations are plotted in the left and right panels of Figure~\ref{fig:2}, together with the DPD$(\widehat{\alpha})$ and GAM-estimates of their regression functions.

The plots show that the GAM-estimates lack smoothness and are shifted upwards in relation to the DPD$(\widehat{\alpha})$-estimates. Both facts are attributable to a lack of resistance of GAM-estimates towards the considerable number of patients with atypically lengthy hospital stays given their age. It should be noted that, while it is always possible to manually increase the smoothness of the GAM-estimates, non-robust automatic methods are very often affected by outlying observations resulting in under or oversmoothed estimates, as observed by \citet{Cantoni:2001a}. Thus, in practice, robust methods of estimation need to be accompanied by robust model selection criteria.

On the other hand, our algorithm selects $\alpha = 1$ in both cases which combined with the robust AIC proposed in Section~\ref{sec:4} leads to reliable estimates for the regression functions, even in the presence of numerous outlying observations. These estimates largely conform to our intuition, as they suggest that older patients are, on average, more likely to experience longer hospital stays. To detect the outlying observations, we may again use the Anscombe residuals of the DPD($\widehat{\alpha})$-estimates, which in the Poisson case are given by
\begin{align*}
r_{A,i} = \frac{\frac{3}{2}(Y_i^{2/3}-\widehat{\mu}_i^{2/3})}{\widehat{\mu}_i^{1/6}}, \quad (i=1, \ldots, n).
\end{align*} 
Observations with $|r_{A,i}| \geq 2.6$ are indicated with a different shape and colour coding in Figure~\ref{fig:2}. These panels suggest that while there exist patients with atypically brief stays, the vast majority of outliers is in the opposite direction, thereby explaining the upper vertical shift of the sensitive GAM-estimates.

\section{Discussion}
\label{sec:7}

This paper greatly extends penalized likelihood methods for nonparametric estimation in GLMs and derives new and important theoretical properties for this broad class of estimators. In practice, the proposed class of estimators behaves similarly to non-robust GAM estimators in the absence of atypical observations, but exhibits a high degree of robustness in their presence. These properties make the proposed methodology particularly well-suited for the analysis of many complex datasets commonly encountered nowadays, such as the diabetes and length of hospital stay data analysed in Section~\ref{sec:6}.

There is a number of interesting and practically useful extensions we aim to consider in future work. These include the case of higher-dimensional non-parametric components, modelled, for example, with thin-plate or tensor product penalties \citep[Chapter 5]{Wood:2017}, as well as more general semi-parametric models based on density power divergence that would allow for both parametric and non-parametric components. Currently, our density power divergence estimator depends on the tuning parameter $\alpha$ and for the selection of $\alpha$ we have developed a data-dependent scheme. An intriguing alternative would be a composite penalized estimator involving several values of $\alpha$, as proposed by  \citet{Zou:2008} in the context of quantile regression. Such an approach has the potential of producing another resistant yet highly efficient competitor to standard maximum likelihood estimators.

\section*{Acknowledgements}

We thank Professor Alfio Marazzi (Lausanne University Hospital) for providing the length of hospital stay data. The research of I. Kalogridis was supported by the Research Foundation-Flanders (project 1221122N). Their support is gratefully acknowledged. G. Claeskens acknowledges support from the KU Leuven Research Fund C1-project C16/20/002.

\section{Appendix: proofs of the theoretical results}

\section{Proofs of Theorem 1 and Corollary 1}

Our main proofs are based on the following optimization lemma regarding extrema in real Hilbert spaces.

\begin{lemma}
\label{lem:1}
Let $\mathcal{H}$ denote a real Hilbert space of functions with inner product $\langle \cdot , \cdot \rangle_{\mathcal{H}}$ and associated norm $\|\cdot\|_{\mathcal{H}}$ and let $L: \mathcal{H} \to \mathbbm{R}$ denote a weakly lower semi-continuous functional whose range is bounded from below, say, $ L: \mathcal{H} \to [M,\infty)$. If there exists $g$ such that $\|g\|_{\mathcal{H}} < 1$ and
\begin{align*}
L(g)< \inf_{\|f\|_{\mathcal{H}} = 1} L(f),
\end{align*}
then $L$ possesses a (possibly) local minimum $f_0$ in the interior of the ball $\{f \in \mathcal{H}: ||f||_{\mathcal{H}} \leq 1 \}$.
\end{lemma}

\begin{proof}

Define $\mathcal{B} = \{f \in \mathcal{H}: ||f||_{\mathcal{H}} \leq 1 \}$ and set $y = \inf_{f \in \mathcal{B}} L(f)$. Observe that $y$ is finite, as the range of $L$ is bounded from below. Thus, there exists a minimizing sequence $\{f_n\}_n \in \mathcal{B}$, that is,
\begin{align*}
\lim_{n \to \infty} L(f_n) =  \inf_{f \in \mathcal{B}} L(f) = y.
\end{align*}
Now, the ball $\mathcal{B}$ is closed, bounded and convex. The space $\mathcal{H}$ is reflexive, \citep[see, e.g.,][]{Rynne:2008}, hence $\mathcal{B}$ is weakly compact. Therefore, there exists a subsequence $\{f_{n_k}\}_k$, which converges weakly to some $f_0 \in \mathcal{B}$. The weak lower semicontinuity of $L$ now implies
\begin{align*}
y \leq L(f_0) = L(\lim_{k \to \infty} f_{n_k}) \leq \liminf_{k \to \infty} L(f_{n_k}) = y,
\end{align*}
and it must be that $L(f_0) = y$. Our assumptions then yield $L(f_0) \leq L(g) < \inf_{\|f\|_{\mathcal{H}} = 1} L(f)$, which further implies that $f_0$ is in the interior of the ball, i.e., $\|f_0\|_{\mathcal{H}} <1$. The proof is complete.
\end{proof}

The following lemma will allow us to compare sums with integrals and may be viewed as a version of the Euler-Maclaurin formula.

\begin{lemma}[Quadrature]
\label{lem:2}
Let $m \geq 1$. Assuming that the design is quasi-uniform in the sense of (A6), there exists a constant $c_m$ depending only on $m$ such that, for all $f\in \mathcal{W}^{m,2}([0,1])$ and all $n \geq 2$,
\begin{equation*}
\|f\|^2 - \frac{c_m}{n \lambda^{1/2m}} \|f\|_{m, \lambda}^2 \leq \frac{1}{n} \sum_{i=1}^n |f(t_{i})|^2 \leq \|f\|^2 + \frac{c_m}{n \lambda^{1/2m}} \|f\|_{m, \lambda}^2.
\end{equation*}
\end{lemma}
\begin{proof}
The proof is given in \citet[Lemma 2.27, Chapter 13]{Eg:2009} with $h$ in their notation equivalent to $\lambda^{1/2m}$ in ours. For asymptotically quasi-uniform designs the inequalities hold for $n \geq n_0$.
\end{proof}

We now tend to the proof of Theorem~1. For ease of notation we shall henceforth denote all generic positive constants with $c_0$. Thus, the value of $c_0$ may change from appearance to appearance.

\begin{proof}[Proof of Theorem~1]

Let us use $L_n(g, \widehat{\alpha}_n)$ to denote the objective function, i.e.,
\begin{align*}
L_n(g, \widehat{\alpha}_n) = \frac{1}{n} \sum_{i=1}^n l_{\widehat{\alpha}_n}(Y_i, g(t_i)) + \lambda \|g^{(m)}\|^2,
\end{align*}
and, as in the text, denote the true function with $g_0(t)$.  Notice that by (A2), for every $\delta>0$ we have $|\widehat{\alpha}_n-\alpha_0|<\delta$ with high probability for all large $n$. Choose small enough $\delta>0$ satisfying (A3)--(A5) such that $\alpha_0-\delta>0$ and observe that
\begin{align}
\label{eq:A1}
\Pr&\left[ \inf_{\|g\|_{m,\lambda} = D} L_n(g_0 + C_n^{1/2}g, \widehat{\alpha}_n) > L_n(g_0, \widehat{\alpha}_n) \right] \nonumber = 
\\ & \Pr\left[ \inf_{\|g\|_{m,\lambda} = D} L_n(g_0 + C_n^{1/2}g, \widehat{\alpha}_n) > L_n(g_0, \widehat{\alpha}_n), |\widehat{\alpha}_n-\alpha_0|<\delta/2 \right] + o(1),
\end{align}
as $n \to \infty$. We will show that for every $\epsilon>0$ there exists a sufficiently large $ D= D_{\epsilon}$ such that
\begin{align}
\label{eq:A2}
\liminf_{n \to \infty}\Pr\left[ \inf_{\|g\|_{m,\lambda} = D} L_n(g_0 + C_n^{1/2}g, \widehat{\alpha}_n) > L_n(g_0, \widehat{\alpha}_n), |\widehat{\alpha}_n-\alpha_0|<\delta/2 \right] \geq 1-\epsilon/2,
\end{align}
where $C_n = n^{-1} \lambda^{-1/2m} + \lambda$. Provided that we can check the conditions of Lemma~\ref{lem:1}, \eqref{eq:A1} and \eqref{eq:A2} together would imply the existence of a (local) minimizer in the ball $\{f \in \mathcal{W}^{m,2}([0,1]):\|f-g_0\|_{m,\lambda} \leq D C_n^{1/2} \}$ with probability at least $1-\epsilon$, which in turn we would establish the result of Theorem~1.

To check the conditions of Lemma~\ref{lem:1} we need to check the weak lower semicontinuity of $L_n(g, \widehat{\alpha}_n)$ for $\widehat{\alpha}_n > \alpha_0 -\delta>0$, as, by the uniform boundedness of the densities given in (A3), $L_n(g, \widehat{\alpha}_n)$ is bounded from below. Let $g_k \to g$ weakly in $\mathcal{W}^{m,2}([0,1])$ as $k \to \infty$ and let $\mathcal{R}_{m,\lambda}$ denote the RK of $\mathcal{W}^{m,2}([0,1])$ for the chosen inner product. The reproducing property and the definition of weak convergence imply
\begin{align*}
\lim_{k\to \infty} g_k(t_i) = \lim_{k\to \infty} \langle g_k, \mathcal{R}_{m,\lambda}(\cdot, t_i) \rangle_{m,\lambda} = \langle g, \mathcal{R}_{m,\lambda}(\cdot, t_i) \rangle_{m,\lambda} = g(t_i), \quad (i = 1, \ldots, n).
\end{align*}
At the same time, by the Hahn-Banach theorem \citep{Rynne:2008}, norms are weakly lower semicontinuous and therefore $\|g^{(m)}\|^2 \leq  \liminf_{k \to \infty}\|g_k^{(m)}\|^2$. Combining these two observations yields
\begin{align*}
L_n(g, \widehat{\alpha}_n)  \leq \liminf_{k \to \infty} L_n(g_k, \widehat{\alpha}_n),
\end{align*}
which is equivalent to weak lower semicontinuity.

To establish \eqref{eq:A2}, use the fundamental theorem of calculus to decompose the difference $L_n(g_0 + C_n^{1/2}g, \widehat{\alpha}_n) - L_n(g_0, \widehat{\alpha}_n)$ as follows:
\begin{align*}
L_n(g_0+C_n^{1/2}g, \widehat{\alpha}_n) - L_n(g_0, \widehat{\alpha}_n) &= \frac{1}{n} \sum_{i=1}^n l_{\widehat{\alpha}_n}(Y_i, g_0(t_i) + C_n^{1/2}g(t_i)) - \frac{1}{n} \sum_{i=1}^n l_{\widehat{\alpha}_n}(Y_i, g_0(t_i))
\\ & \quad + 2 C_n^{1/2} \lambda \langle g_0^{(m)}, g^{(m)} \rangle + \lambda C_n\|g^{(m)}\|^2
\\ & = \frac{1}{n} \sum_{i=1}^n \int_{0}^{C_n^{1/2}(g(t_i)} \{l_{\widehat{\alpha}_n}^{\prime}(Y_i, g_0(t_i)+u) - l_{\widehat{\alpha}_n}^{\prime}(Y_i, g_0(t_i))\} du
\\ & \quad + \frac{C_n^{1/2}}{n}  \sum_{i=1}^n g(t_i) l_{\widehat{\alpha}_n}^{\prime}(Y_i, g_0(t_i)) + 2 C_n^{1/2} \lambda \langle g_0^{(m)}, g^{(m)} \rangle
\\ &\quad + \lambda C_n\|g^{(m)}\|^2
\\ & = I_1(g, \widehat{\alpha}_n) + I_2(g, \widehat{\alpha}_n) + I_3(g, \widehat{\alpha}_n),
\end{align*}
say, with 
\begin{align*}
I_1(g, \widehat{\alpha}_n) &:= \frac{1}{n} \sum_{i=1}^n \int_{0}^{C_n^{1/2}g(t_i)} \{l_{\widehat{\alpha}_n}^{\prime}(Y_i, g_0(t_i)+u) - l_{\widehat{\alpha}_n}^{\prime}(Y_i, g_0(t_i))\} du + \lambda C_n \|g^{(m)}\|^2
\\ I_2(g, \widehat{\alpha}_n) &:=  \frac{C_n^{1/2}}{n}  \sum_{i=1}^n g(t_i) l_{\widehat{\alpha}_n}^{\prime}(Y_i, g_0(t_i))
\\ I_3(g, \widehat{\alpha}_n) &=  2 C_n^{1/2} \lambda \langle g_0^{(m)}, g^{(m)} \rangle.
\end{align*}
Here and in the sequel we write
\begin{align*}
l^{\prime}_{\alpha}(Y_i,g(t_i)) = \frac{\partial l_{\alpha}(Y_i,x)}{\partial x} \bigg \vert_{x = g(t_i)}
\end{align*}
and define higher order derivatives in the same manner.

Our proof consists of showing the following
\begin{align}
\inf_{\|g\|_{m, \lambda} = D,\ |\alpha - \alpha_0| < \delta/2} \mathbb{E}\{I_1(g, \alpha) \} & \geq c_0 D^2 C_n \label{eq:A3} \\
\sup_{\|g\|_{m, \lambda} \leq D, \ |\alpha - \alpha_0| < \delta/2	 } |I_1(g, \alpha)-\mathbb{E}\{I_1(g, \alpha)\}| & = o_P(1) C_n \label{eq:A4}
\\ \sup_{\|g\|_{m, \lambda} \leq D} |I_2(g, \widehat{\alpha}_n)| &= O_P(1) D C_n \label{eq:A5} \\
\sup_{\|g\|_{m, \lambda} \leq D} |I_3(g, \widehat{\alpha}_n)| & = O(1) D C_n \label{eq:A6},
\end{align}
for a strictly positive $c_0$ in \eqref{eq:A3}. In combination, \eqref{eq:A3}--\eqref{eq:A6} would imply that for sufficiently large $D$, $\mathbb{E}\{I_1(g, \widehat{\alpha}_n) \}$ would be positive at the $D$-sphere and dominate all other terms. Thus, \eqref{eq:A2} would hold and the theorem would be proven.

We begin by showing \eqref{eq:A6}. For this, observe that, by the Schwarz inequality in $L^2([0,1])$, we have
\begin{align*}
|I_3(g, \widehat{\alpha}_n)| \leq 2 C_n^{1/2} \lambda \| g_0^{(m)}\| \|g^{(m)}\| = 2 C_n^{1/2} \lambda^{1/2} \|g_0^{(m)}\| \lambda^{1/2} \|g^{(m)}\|,
\end{align*}
whence, by definition of $C_n$,
\begin{align*}
\sup_{\|g\|_{m,\lambda} \leq D} |I_3(g, \widehat{\alpha}_n)| \leq 2 D \lambda^{1/2} C_n^{1/2} \|g_0^{(m)}\| \leq c_0 D C_n.
\end{align*}
The bound in \eqref{eq:A6} now follows.

We now establish the bound in \eqref{eq:A5}. Clearly,
\begin{align}
\label{eq:A7}
\sup_{\|g\|_{m,\lambda} \leq D} \left| \frac{1}{n} \sum_{i=1}^n g(t_i) l^{\prime}_{\widehat{\alpha}_n}(Y_i,g_0(t_i)) \right| \leq \sup_{\|g\|_{m,\lambda} \leq D, |\alpha-\alpha_0|<\delta/2} \left| \frac{1}{n} \sum_{i=1}^n g(t_i) l^{\prime}_{\alpha}(Y_i,g_0(t_i)) \right|.
\end{align}
Fix $\alpha \in (\alpha_0-\delta/2, \alpha_0+\delta/2)$ for the time being and consider the supremum of the RHS in \eqref{eq:A7} over $\mathcal{B}_D:=\{g \in \mathcal{W}^{m,2}([0,1]):\|g\|_{m,\lambda} \leq D \}$. The random variables $l^{\prime}_{\alpha}(Y_i,g_0(t_i)), i = 1, \ldots, n$ are independent and, by Fisher consistency, have mean zero for every $\alpha>0$. Moreover, by (A3), they are uniformly bounded. Hence, they are also uniformly sub-Gaussian. Let $Q_n$ denote the empirical measure of the $t_i$, that is, 
\begin{align*}
Q_n = \frac{1}{n} \sum_{i=1}^n \delta_{t_i}.
\end{align*}
Further, let $\mathcal{H}(\epsilon, \mathcal{B}_D, Q_n)$ denote the $\epsilon$-entropy in the semi-norm
\begin{align*}
\|g\|_{Q_n}^2 = \int |g|^2 Q_n = \frac{1}{n} \sum_{i=1}^n |g(t_i)|^2.
\end{align*}
That is, $\mathcal{H}(\epsilon, \mathcal{B}_D, Q_n)$ is the logarithm of the smallest value of $N$ such that there exists $\{g_j\}_{j=1}^N$ with the property that
\begin{align*}
\sup_{g \in \mathcal{B}_D} \min_{j=1, \ldots, N} \|g-g_j\|_{Q_n} \leq \epsilon.
\end{align*}
We will now bound the entropy  $\mathcal{H}(\epsilon, \mathcal{B}_D, Q_n)$. Notice first that as
\begin{align*}
\|f-g\|_{Q_n} \leq \|f - g\|_{\infty},
\end{align*}
we have $\mathcal{H}(\epsilon, \mathcal{B}_D, Q_n) \leq \mathcal{H}_{\infty}(\epsilon, \mathcal{B}_D)$ for all $n$, where the latter stands for the entropy in the supremum norm. Hence, it suffices to bound $\mathcal{H}_{\infty}(\epsilon, \mathcal{B}_D)$. Notice next that, by assumption, $\lambda \in (0,1]$ for all large $n$, hence
\begin{align*}
\|g\|_{m,\lambda} \geq \lambda^{1/2} \|g\|_{m,1},
\end{align*}
so that
\begin{align*}
\mathcal{B}_D \subset \{g \in \mathcal{W}^{m,2}([0,1]): \|g\|_{m,1} \leq \lambda^{-1/2} D \},
\end{align*}
which is the  $D$-ball in $\mathcal{W}^{m,2}([0,1])$ equipped with its standard norm. Consequently, by Proposition 6 of \citet{Cucker:2001},
\begin{align}
\label{eq:A8}
\mathcal{H}(\epsilon, \mathcal{B}_D, Q_n) \leq H_{\infty}(\epsilon, \mathcal{B}_D) \leq C D^{1/m} \lambda^{-1/2m} \epsilon^{-1/m},
\end{align}
for some universal constant $C>0$ and all $\epsilon>0$. 

We aim to apply Corollary 8.3 of \citet{van de Geer:2000} and to that end note that by Lemma~\ref{lem:2} and our limit assumptions which entail $n \lambda^{1/2m} \to \infty$, we have
\begin{align*}
\sup_{g \in \mathcal{B}_D} \left| \frac{1}{n} \sum_{i=1}^n |g(t_i)|^2 \right| \leq c_0^2 D^2,
\end{align*}
for some $c_0>0$ depending only on $m$. Take the square root of both sides of \eqref{eq:A8} and integrate from $0$ to $c_0D$ to arrive at
\begin{align*}
\int_{0}^{c_0D} \mathcal{H}^{1/2}(u, \mathcal{B}_D, Q_n) du  \leq C^{1/2} c_0 \frac{2m}{2m-1} D^{1/2m} \lambda^{-1/4m} D^{1-1/2m} = c_0 \lambda^{-1/4m} D,
\end{align*}
for some $c_0>0$. Corollary 8.3 of \citet{van de Geer:2000} now applies and gives
\begin{align}
\label{eq:A9}
\Pr\left[ \sup_{g \in \mathcal{B}_D} \left| \frac{1}{n} \sum_{i=1}^n g(t_i) l^{\prime}_{\alpha}(Y_i,g_0(t_i)) \right| \geq  \frac{c_0 D}{2 n^{1/2} \lambda^{1/4m}} \right] \leq c_0 \exp\left[ -c_0 \lambda^{-1/2m} \right],
\end{align}
for some strictly positive $c_0$.

To make the above argument uniform in $\alpha \in (\alpha_0-\delta/2, \alpha_0+\delta/2)$ partition this interval into $N^{\prime}$ intervals with radii no larger than $c_0 D n^{-1/2} \lambda^{-1/4m}/2$. Notice that since $n^{-1/2} \lambda^{-1/4m} \to 0$, the radii can be assumed smaller than $\delta/4$ for all large $n$. Select an $\alpha_k$ in each one of these intervals. Clearly,
\begin{align*}
\sup_{\substack{g \in \mathcal{B}_D \\ |\alpha-\alpha_0|<\delta/2}} \left| \frac{1}{n} \sum_{i=1}^n g(t_i) l^{\prime}_{\alpha}(Y_i,g_0(t_i)) \right|&  \leq \max_{ k \leq N^{\prime}} \sup_{g \in \mathcal{B}_D} \left| \frac{1}{n} \sum_{i=1}^n g(t_i) l^{\prime}_{\alpha_k}(Y_i,g_0(t_i)) \right| \\ & \quad + \sup_{|\alpha-\alpha_0|<\delta/2} \min_{k \leq N^{\prime}} \sup_{g \in \mathcal{B}_D} \left| \frac{1}{n} \sum_{i=1}^n g(t_i)( l^{\prime}_{\alpha}(Y_i,g_0(t_i)) -l^{\prime}_{\alpha_k}(Y_i,g_0(t_i))) \right|.
\end{align*}
Now, by the triangle and Schwarz inequalities,
\begin{align*}
\left|\frac{1}{n} \sum_{i=1}^n g(t_i)( l^{\prime}_{\alpha}(Y_i,g_0(t_i)) -l^{\prime}_{\alpha_k}(Y_i,g_0(t_i)))\right| &\leq \frac{1}{n} \sum_{i=1}^n |g(t_i)|| l^{\prime}_{\alpha}(Y_i,g_0(t_i)) -l^{\prime}_{\alpha_k}(Y_i,g_0(t_i))|
 \\ & \leq \max_{i \leq n} |l^{\prime}_{\alpha}(Y_i, g_0(t_i))- l^{\prime}_{\alpha_k}(Y_i, g_0(t_i))| \left\{\frac{1}{n} \sum_{i=1}^n |g(t_i)|^2 \right\}^{1/2}
\\ & \leq c_0 D \max_{i \leq n} |l^{\prime}_{\alpha}(Y_i, g_0(t_i))- l^{\prime}_{\alpha_k}(Y_i, g_0(t_i))|,
\end{align*}
for all $g \in \mathcal{B}_D$, by Lemma~\ref{lem:2}. By the triangle inequality yet again,
\begin{align*}
|l^{\prime}_{\alpha}(Y_i, g_0(t_i))- l^{\prime}_{\alpha_k}(Y_i, g_0(t_i))| &\leq (1+\alpha) \left|  \int_{\mathbbm{R}} \left\{ f_{\theta_{0,i}}^{1+\alpha}(y) - f_{\theta_{0,i}}^{1+\alpha_k}(y) \right\} u_{\theta_{0,i}}(y) dy  \right| 
 \\ & \quad + \left|\alpha_k-\alpha\right| \left| \int_{\mathbbm{R}} f_{\theta_{0,i}}^{1+\alpha_k}(y) u_{\theta_{0,i}}(y) dy \right| + \left| \alpha_k - \alpha\right| \left|  f_{\theta_{0,i}}^{\alpha}(Y_i) u_{\theta_{0,i}}(Y_i) \right|
 \\ & \quad +  (1+\alpha)\left| f_{\theta_{0,i}}^{\alpha_k}(Y_i) - f_{\theta_{0,i}}^{\alpha}(Y_i) \right| \left|u_{\theta_{0,i}}(Y_i) \right|.
 \end{align*}
Whenever $|\alpha -\alpha_k|<\delta/4$, by the mean-value theorem, for every $y \in \mathcal{Y}$ there exists an $\widetilde{\alpha}_k = \widetilde{\alpha}_k(y)$ such that $|\widetilde{\alpha}_k - \alpha|  \leq |\alpha_k - \alpha |	 < \delta/4$ and 
\begin{align*}	
\left|f_{\theta_{0,i}}^{\alpha_k}(y) - f_{\theta_{0,i}}^{\alpha}(y)\right| = \left| \alpha_k- \alpha\right| \left| f_{\theta_{0,i}}^{\widetilde{\alpha}_k}(y) \log\left(f_{\theta_{0,i}}(y) \right)  \right|.
\end{align*}
Furthermore, $|\widetilde{\alpha}_k-\alpha_0| \leq |\widetilde{\alpha}_k-\alpha| + |\alpha-\alpha_0|   < \delta/4 + \delta/2 < \delta$. Therefore,
\begin{align*}
|l^{\prime}_{\alpha}(Y_i, g_0(t_i))- l^{\prime}_{\alpha_k}(Y_i, g_0(t_i))| & \leq (1+\alpha) \left| \alpha_k - \alpha\right|   \mathbb{E} \left\{ \sup_{|\alpha-\alpha_0|<\delta}  \left| f_{\theta_{0,i}}^{\alpha}(Y_i) \log\left( f_{\theta_{0,i}}(Y_i) \right)  u_{\theta_{0,i}}(Y_i) \right| \right\} \nonumber
\\ & \quad + \left| \alpha_k - \alpha\right| \mathbb{E}\left\{ \left| f_{\theta_{0,i}}^{\alpha_k}(Y_i) u_{\theta_{0,i}}(Y_i) \right| \right\} + \left| \alpha_k - \alpha\right| \sup_{y \in \mathcal{Y}} \left|  f_{\theta_{0,i}}^{\alpha}(y) u_{\theta_{0,i}}(y) \right| \nonumber
\\ & \quad + (1+\alpha)\left| \alpha_k - \alpha\right| \sup_{y \in \mathcal{Y}} \sup_{|\alpha-\alpha_0|<\delta} \left|  f_{\theta_{0,i}}^{\alpha}(y) \log\left( f_{\theta_{0,i}}(y) \right)  u_{\theta_{0,i}}(y) \right| \nonumber
\\ & \leq c_0 \left| \alpha_k - \alpha\right|,
\end{align*}
By (A3), $c_0$ does not depend on $i$, $n$, $\alpha$ or $\alpha_k$. Hence, since the radii are smaller than $c_0 D n^{-1/2} \lambda^{-1/4m}/2$, we have
\begin{align*}
\Pr & \left[  \sup_{\substack{g \in \mathcal{B}_D \\ |\alpha-\alpha_0|<\delta/2}} \left| \frac{1}{n} \sum_{i=1}^n g(t_i) l^{\prime}_{\alpha}(Y_i,g_0(t_i)) \right| \geq \frac{c_0 D}{n^{1/2} \lambda^{1/4m}} \right] 
\\ & \leq \Pr\left[ \bigcup_{k=1}^{N^{\prime}} \left\{  \sup_{g \in \mathcal{B}_D} \left| \frac{1}{n} \sum_{i=1}^n g(t_i) l^{\prime}_{\alpha_k}(Y_i,g_0(t_i))\right|  \geq \frac{c_0 D}{2 n^{1/2} \lambda^{1/4m}} \right\}  \right]
\\ & \leq \sum_{k=1}^{N^{\prime}} \Pr\left[   \sup_{g \in \mathcal{B}_D} \left| \frac{1}{n} \sum_{i=1}^n g(t_i) l^{\prime}_{\alpha_k}(Y_i,g_0(t_i))\right|  \geq \frac{c_0 D}{2 n^{1/2} \lambda^{1/4m}}  \right]
\\ & \leq N^{\prime} c_0 \exp\left[-c_0 \lambda^{-1/2m}\right] \leq c_0  \frac{n^{1/2} \lambda^{1/4m}}{D}  \exp\left[-c_0 \lambda^{-1/2m}\right],
\end{align*}
where in the second-to-last step we have applied \eqref{eq:A9} and in the last step we have bounded $N^{\prime}$ using Lemma 2.5 of \citep{van de Geer:2000}. Our limit assumptions now imply that 
\begin{align}
\label{eq:A10}
C_n^{1/2}\sup_{\substack{g \in \mathcal{B}_D \\ |\alpha-\alpha_0|<\delta/2}} \left| \frac{1}{n} \sum_{i=1}^n g(t_i) l^{\prime}_{\alpha}(Y_i,g_0(t_i)) \right| = O_P(1) D C_n^{1/2} n^{-1/2} \lambda^{-1/4m} = O_P(1) D C_n,
\end{align}
by definition of $C_n$. Equations \eqref{eq:A7} and \eqref{eq:A10} jointly imply \eqref{eq:A5}.

We now establish the lower bound in \eqref{eq:A3}. Let us begin by noting that the embedding (7) in the main text implies the existence of a symmetric function, $\mathcal{R}_{m,\lambda}: [0,1]^2 \to \mathbbm{R}$, the reproducing kernel, such that for every $y \in [0,1]$ the map $x \mapsto \mathcal{R}_{m,\lambda}(x,y) \in \mathcal{W}^{m,2}([0,1])$ and
\begin{align*}
f(x) = \langle \mathcal{R}_{m,\lambda}(x,\cdot), f \rangle_{m,\lambda}, \quad f \in \mathcal{W}^{m,2}([0,1]).
\end{align*}
Using these facts, the Schwarz inequality in $\mathcal{W}^{m,2}([0,1])$ and (7) in the main text, we obtain
\begin{equation*}
\|\mathcal{R}_{m, \lambda}(x, \cdot)\|_{m, \lambda}^2 =  \langle \mathcal{R}_{m, \lambda}(x, \cdot), \mathcal{R}_{m, \lambda}(\cdot, x) \rangle_{m,\lambda} =  \mathcal{R}_{m, \lambda}(x, x)   \leq c_0 \lambda^{-1/{4m}}  \|\mathcal{R}_{m, \lambda}(x, \cdot)\|_{m, \lambda}.
\end{equation*}
Divide both sides by $\|\mathcal{R}_{m, \lambda}(x, \cdot)\|_{m, \lambda}$ to get
\begin{align}
\label{eq:A11}
\sup_{x\in [0,1]} \|\mathcal{R}_{m, \lambda}(x, \cdot)\|_{m, \lambda} \leq c_0 \lambda^{-1/4m}.
\end{align}
Since $\lambda>0$, it is clear that this inequality also holds when $\|\mathcal{R}_{m, \lambda}(x, \cdot)\|_{m, \lambda} = 0$. Hence, by the reproducing property, the Schwarz inequality and \eqref{eq:A11}, we have
\begin{align}
\label{eq:A12}
C_n^{1/2}\sup_{\|g\|_{m,\lambda} \leq D} \max_{i \leq n} |g(t_i)| & = C_n^{1/2}\sup_{\|g\|_{m,\lambda} \leq D} \max_{i \leq n} |\langle \mathcal{R}_{m,\lambda}(t_i, \cdot), g \rangle_{m,\lambda}| \nonumber
\\ & \leq  D C_n^{1/2} \sup_{x \in [0,1]} \|\mathcal{R}_{m,\lambda}(x, \cdot)\|_{m,\lambda} \nonumber
\\ & \leq c_0 D C_n^{1/2} \lambda^{-1/4m} \nonumber
\\ & = o(1),
\end{align}
for every fixed $D$, as $n \to \infty$. Now, since $|\widehat{\alpha}_n - \alpha_0| < \delta/2$, we clearly have
\begin{align*}
\mathbb{E} \left\{ \int_{0}^{C_n^{1/2}g(t_i)} \{l_{\widehat{\alpha}_n}^{\prime}(Y_i, g_0(t_i)+u) - l_{\widehat{\alpha}_n}^{\prime}(Y_i, g_0(t_i))\} du \right\} \geq  \inf_{|\alpha-\alpha_0|<\delta/2} \left[  \int_{0}^{C_n^{1/2}g(t_i)} \mathbb{E}\{l_{\alpha}^{\prime}(Y_i, g_0(t_i)+u)\} du \right],
\end{align*}
where we have used the fact that $\mathbb{E}\{l_{\alpha}^{\prime}(Y_i, g_0(t_i)) = 0$ for any fixed $\alpha >0$, by Fisher consistency. Furthermore,
\begin{align*}
\int_{0}^{C_n^{1/2}g(t_i)} \mathbb{E}\{l_{\alpha}^{\prime}(Y_i, g_0(t_i)+u)\} du & = \mathbb{E}\{l_{\alpha}(Y_i, g_0(t_i)+C_n^{1/2}g(t_i))\}- \mathbb{E}\{l_{\alpha}(Y_i, g_0(t_i)\}
\\ & = C_n^{1/2}g(t_i) \mathbb{E}\{l_{\alpha}^{\prime}(Y_i, g_0(t_i)\} 
\\ & \quad + 2^{-1}C_n|g(t_i)|^2  \mathbb{E}\{l_{\alpha}^{\prime \prime}(Y_i, g_0(t_i) + s_i C_n^{1/2}g(t_i)\}
\\ & =  2^{-1}C_n|g(t_i)|^2  \mathbb{E}\{l_{\alpha}^{\prime}(Y_i, g_0(t_i) + s_i C_n^{1/2}g(t_i)\}
\end{align*}
for some value $s_i$ satisfying $|s_i| \leq 1$. By the definitions in (A4),
\begin{align*}
\mathbb{E}\{l_{\alpha}^{\prime \prime}(Y_i, g_0(t_i)+ C_n^{1/2}s_ig(t_i))\} = m_{t_i}(C_n^{1/2}s_ig(t_i), \alpha) - \left(1+\frac{1}{\alpha}\right) \mathbb{E}\{n_{t_i}(s_iC_n^{1/2}g(t_i), \alpha, Y_i)\},
\end{align*}
for equicontinuous functions $m_{t_i}(u,\alpha)$ and $n_{t_i}(u, \alpha, y)$ at $u = 0$. By the local uniform boundedness of $n_{t_i}(u, \alpha, y)$ in assumption (A4) and dominated convergence, it may be verified that the functions $\{ \mathbb{E}\{n_{t}(u, \alpha, Y_i)\}, t \in [0,1] \}$ are also equicontinuous. Observe next that
\begin{align*}
\mathbb{E}\{l_{\alpha}^{\prime \prime}(Y_i, g_0(t_i)\} = \mathbb{E}\{f_{\theta_{0,i}}^{\alpha}(Y_i)|u_{\theta_{0,i}}(Y_i)|^2 \},
\end{align*}
so that, by (A5), there exists a $c_0>0$ such that
\begin{align*}
\mathbb{E}\{l_{\alpha}^{\prime \prime}(Y_i, g_0(t_i)\} \geq \inf_{|\alpha-\alpha_0|<\delta/2} \inf_{n} \min_{1\leq i \leq n} \mathbb{E}\{ f_{\theta_{0,i}}^{\alpha}(Y_i) |u_{\theta_{0,i}}(Y_i)|^2 \} \geq c_0>0.
\end{align*}
By equicontinuity and the fact that, by \eqref{eq:A12}, for every $\|g\|_{m,\lambda} \leq D$ and $t \in [0,1]$, $C_n^{1/2}|g(t)| \to 0$, as $n \to \infty$, conclude that for small enough $\delta>0$ and all large $n$,
\begin{align*}
\mathbb{E}\{l_{\alpha}^{\prime \prime}(Y_i, g_0(t_i)+ C_n^{1/2}s_ig(t_i))\} \geq \inf_{|u| <\delta} \mathbb{E}\{l_{\alpha}^{\prime \prime}(Y_i, g_0(t_i)+u) \}  \geq c_0,
\end{align*}
for some $c_0$ not depending on either $\alpha$ or $i$. Thus, there exists a $c_0>0$ such that
\begin{align*}
\mathbb{E} \left\{ \int_{0}^{C_n^{1/2}g(t_i)} \{l_{\widehat{\alpha}_n}^{\prime}(Y_i, g_0(t_i)+u) - l_{\widehat{\alpha}_n}^{\prime}(Y_i, g_0(t_i))\} du \right\} \geq c_0 C_n |g(t_i)|^2,
\end{align*}
independently of $i$, $n$ and $\alpha$. Averaging and approximating the sum from below with the help of Lemma~\ref{lem:2} we now see that
\begin{align*}
\inf_{|\alpha-\alpha_0|<\delta/2}\mathbb{E}\{I_1(g, \alpha) \}  & \geq c_0 \frac{C_n}{n} \sum_{i=1}^n |g(t_i)|
^2 + \lambda C_n \|g^{(m)}\|^2 ,
\\ & \geq \min\{1,c_0\} C_n \|g\|_{m,\lambda}^2\left(1 - \frac{c_m}{n \lambda^{1/2m}} \right)
\\ & \geq c_0 C_n \|g\|_{m,\lambda}^2,
\end{align*}
for some $c_0>0$, where we have used our limit assumptions, which entail that $n^{2m} \lambda \to \infty$. Taking the infimum over the $D$-sphere and iterating infima,
\begin{align*}
\inf_{\|g\|_{m,\lambda} = D,\ |\alpha-\alpha_0|<\delta/2}\mathbb{E}\{I_1(g, \alpha) \}  \geq c_0 D^2 C_n,
\end{align*}
for a strictly positive $c_0$, which is precisely \eqref{eq:A3}.

To complete the proof we now show that the remainder term $I_1(g, \alpha) - \mathbb{E}\{I_1(g, \alpha) \}$ is under our assumptions asymptotically negligible uniformly in $\|g\|_{m,\lambda} \leq D$ and $\alpha \in (\alpha_0-\delta, \alpha_0+\delta)$. We show this with an empirical process argument. Recall that the $t_i$ are fixed and $Y_i \sim F_{\theta_{0,i}}$. Thus, the distribution of each pair $(t_i, Y_i)$ is given by the product measure $P_i = \delta_{t_i} \times F_{\theta_{0,i}}$. Put $\bar{P} = n^{-1} \sum_{i=1}^n P_i$. Further, let $P_n = n^{-1} \sum_{i=1}^n \delta_{t_i, Y_i}$ denote the empirical measure placing mass $n^{-1}$ on each pair $(t_i, Y_i)$. Then, adopting the notation of \citet[Chapters 5 and 8]{van de Geer:2000} we have
\begin{align*}
I_1(g, \alpha) - \mathbb{E}\{I_1(g, \alpha) \} = \int h_{g, \alpha} d(P_n - \bar{P}) = n^{-1/2} v_n(h_{g, \alpha}),
\end{align*}
where $v_n(\cdot)$ denotes the empirical process and $h_{g,\alpha}$ is the function $[0,1] \times \mathbbm{R} \to \mathbbm{R}$ given by
\begin{align*}
h_{g,\alpha}(t,y) := \int_{0}^{C_n^{1/2}g(t)}\{l^{\prime}_{\alpha}(y, g_0(t)+u) - l^{\prime}_{\alpha}(y, g_0(t)) \} du,
\end{align*}
for each $g \in \mathcal{B}_{D} := \{f \in \mathcal{W}^{m,2}([0,1]): \|f\|_{m,\lambda} \leq D\}$ and $\alpha \in \mathcal{V}_{\alpha_0} :=  (\alpha_0-\delta/2, \alpha_0+\delta/2)$. This class of functions depends on $n$ through $C_n^{1/2}$, but we suppress this dependence for notational convenience.

We will apply Theorem 5.11 of \citet{van de Geer:2000} to this empirical process adapted for independent but not identically distributed random variables $(t_i, Y_i)$, see the remarks in \citet[pp. 131--132]{van de Geer:2000}. For this, we first derive a uniform bound on $h_{g,\alpha}(t,y)$ and a bound on its $\mathcal{L}^2 (\bar{P})$-norm. For the former note that by \eqref{eq:A12} we have $C_n^{1/2} |g(t)| \to 0$ uniformly in $t \in [0,1]$ and $g \in \mathcal{B}_D$. Furthermore, for every $|u| < \delta$ and $\alpha \in \mathcal{V}_{\alpha_0}$ the mean-value theorem reveals that
\begin{align} 
\label{eq:A13}
|l_{\alpha}^{\prime}(y, g_0(t) + u) - l_{\alpha}^{\prime}(y,g_0(t))| \leq \sup_{y \in \mathcal{Y}} \sup_{\substack{ |u| < \delta \\ |\alpha-\alpha_0| <\delta/2 }} | l_{\alpha}^{\prime \prime}(y, \theta_0(t)+u)| \ |u|.
\end{align}
By (A4),
\begin{align}
\sup_{y \in \mathcal{Y}} \sup_{\substack{ |u| < \delta \\ |\alpha-\alpha_0| <\delta/2 }} | l_{\alpha}^{\prime \prime}(y, \theta_0(t)+u)| & \leq \sup_{\substack{ |u|< \delta \\ |\alpha-\alpha_0| <\delta/2 }} \left[ \left|m_{t}(u, \alpha)\right| + \left(1+ \frac{1}{\alpha} \right) \sup_{y \in \mathcal{Y}} \left|n_{t}(u, \alpha, y)\right|  \right] \nonumber
\\ & \leq c_0 \sup_{\substack{ |u|< \delta \\ |\alpha-\alpha_0| <\delta/2 }} \left[ \left|m_{t}(u, \alpha)\right| + \sup_{y \in \mathcal{Y}} \left|n_{t}(u, \alpha, y)\right|  \right] \nonumber
\\ & \leq c_0 M^{\prime} \label{eq:A14},
\end{align}
with $c_0$ and $M^{\prime}$ independent of $t \in [0,1]$. Combining \eqref{eq:A13} and \eqref{eq:A14} yields
\begin{align*}
|h_{g,\alpha}(t,y)| \leq  c_0 M^{\prime} \left| \int_{0}^{C_n^{1/2}g(t)} |u| du \right| \leq 2^{-1}c_0 M^{\prime} C_n |g(t)|^2 \leq c_0 C_n \lambda^{-1/2m},
\end{align*}
for some $c_0$ independent of $g \in \mathcal{B}_D$, $\alpha \in \mathcal{V}_{\alpha_0}$, $t \in [0,1]$ and $y\in \mathcal{Y}$. It follows that we may take $K: = c_0 C_n \lambda^{-1/2m}$ in Lemma 5.8 of \citet{van de Geer:2000}.

Similarly, the Cauchy-Schwarz integral inequality and \eqref{eq:A13}--\eqref{eq:A14} yield
\begin{align*}
\int |h_{g,\alpha}|^2 d \bar{P} & = n^{-1} \sum_{i=1}^n  \mathbb{E} \left\{ \left| \int_{0}^{C_n^{1/2}g(t_i) } \left\{l_{\alpha}^{\prime}(Y_i, g_0(t_i) + u) - l_{\alpha}^{\prime}(Y_i,g_0(t_i)) \right\} du \right|^2 \right\} \nonumber
\\& \leq  c_0 n^{-1} \sum_{i=1}^n |C_n^{1/2}g(t_i)| \left| \int_{0}^{C_n^{1/2}g(t_i) } |u|^2 du \right| \nonumber
\\ & \leq c_0 C_n^{2} \sup_{t\in [0,1]} |g(t)|^2 n^{-1}\sum_{i=1}^n |g(t_i)|^2 \nonumber
\\ & = c_0 \lambda^{-1/{2m}} C_n^{2},
\end{align*}
where in the last step we have used (7) from the main text and Lemma~\ref{lem:2} in order to bound $n^{-1} \sum |g(t_i)|^2$ over $g \in \mathcal{B}_D$. The constant $c_0$ does not depend on $g \in \mathcal{B}_D$, $\alpha \in \mathcal{V}_{\alpha_0}$, $t \in [0,1]$. Thus, we may take $R = c_0 \lambda^{-1/4m} C_n$ in Lemma 5.8 of \citet{van de Geer:2000}. 

With this choice of $K$ and $R$, it follows from Lemma 5.8 of \citet{van de Geer:2000} that $|\rho_{K}(h_{g,\alpha})|^2 \leq c_0  R^2$, where $\rho_{K}(h_{g, \alpha})$ denotes the Bernstein seminorm given by
\begin{align*}
|\rho_K (h_{g, \alpha})|^2 = 2K^2 \int \left( e^{ |h_{g, \alpha}|/K } - 1 - |h_{g, \alpha}|/K \right)  d \bar{P}, \quad K> 0.
\end{align*}
Furthermore, using $\mathcal{H}_{B,K}(\epsilon, \{h_{g, \alpha},\ g \in \mathcal{B}_D, \ \alpha \in \mathcal{V}_{\alpha_0} \}, \bar{P})$ to denote the $\epsilon$-generalized entropy with bracketing in the Bernstein norm $\rho_{K}$, Lemma 5.10 of \citet{van de Geer:2000} shows that
\begin{align*}
\mathcal{H}_{B,K}(\epsilon, \{h_{g, \alpha},\ g \in \mathcal{B}_D, \ \alpha \in \mathcal{V}_{\alpha_0} \}, \bar{P}) \leq \mathcal{H}_B( c_0 \epsilon, \{h_{g, \alpha},\ g \in \mathcal{B}_D, \ \alpha \in \mathcal{V}_{\alpha_0} \}, \bar{P}),
\end{align*}
where $\mathcal{H}_B(\epsilon, \{h_{g, \alpha},\ g \in \mathcal{B}_D, \ \alpha \in \mathcal{V}_{\alpha_0} \}, \bar{P})$ stands for the usual $L_2(\bar{P})$ $\epsilon$-entropy with bracketing. As $\bar{P}$ is a probability measure, we further have
\begin{align*}
\mathcal{H}_B( c_0 \epsilon, \{h_{g, \alpha},\ g \in \mathcal{B}_D, \ \alpha \in \mathcal{V}_{\alpha_0} \}, \bar{P}) \leq \mathcal{H}_{\infty}(c_0\epsilon, \{h_{g, \alpha},\ g \in \mathcal{B}_D, \ \alpha \in \mathcal{V}_{\alpha_0} \}),
\end{align*}
with $H_{\infty}(\epsilon, \mathcal{F})$ denoting the $\epsilon$-entropy in the supremum norm of a class of functions $\mathcal{F}$. 

We next derive a bound for $\mathcal{H}_\infty(\epsilon, \{h_{g, \alpha},\ g \in \mathcal{B}_D, \ \alpha \in \mathcal{V}_{\alpha_0} \})$, for every sufficiently small $\epsilon>0$. For this, observe that, by the triangle inequality,
\begin{align}
\label{eq:A15}
\left|h_{g_1, \alpha_1}(t,y) - h_{g_2,\alpha_2}(t,y)\right| &\leq \left|h_{g_1, \alpha_1}(t,y) - h_{g_2,\alpha_1}(t,y)\right| + \left|h_{g_2, \alpha_1}(t,y) - h_{g_2,\alpha_2}(t,y)\right|  \nonumber
\\ &  \leq \left| \int_{C_n^{1/2}g_2(t)}^{C_n^{1/2}g_1(t) } \left\{l_{\alpha_1}^{\prime}(y, g_0(t) + u) - l_{\alpha_1}^{\prime}(y,g_0(t)) \right\} du \right|  \nonumber
\\ & \quad + \left| \int_{0}^{C_n^{1/2}g_2(t)}  \left\{ l_{\alpha_1}^{\prime}(y, g_0(t)+u) -  l_{\alpha_2}^{\prime}(y, g_0(t)+u) \right\} du \right| \nonumber
\\ & \quad + \left| \int_{0}^{C_n^{1/2}g_2	(t)}  \left\{ l_{\alpha_1}^{\prime}(y, g_0(t)) -  l_{\alpha_2}^{\prime}(y, g_0(t)) \right\} du \right|.
\end{align}
For the first term on the RHS of \eqref{eq:A15}, by \eqref{eq:A13} and \eqref{eq:A14}, we obtain
\begin{align*}
\left| \int_{C_n^{1/2}g_2(t)}^{C_n^{1/2}g_1(t) } \left\{l_{\alpha_1}^{\prime}(y, g_0(t) + u) - l_{\alpha_1}^{\prime}(y,g_0(t)) \right\} du \right| \leq c_0 C_n^{1/2}|g_1(t)- g_2(t)| \leq c_0 C_n^{1/2}\|g_1- g_2\|_{\infty},
\end{align*}
for some $c_0$ not depending on either $y$, $t$ or $\alpha$. Next, for the second term we have, by the mean value theorem and (A3),
\begin{align*}
&\left| \int_{0}^{C_n^{1/2}g_2(t)}  \left\{ l_{\alpha_1}^{\prime}(y, g_0(t)+u) -  l_{\alpha_2}^{\prime}(y, g_0(t)+u) \right\} du \right|
\\ & \quad \leq 2 c_0 C_n^{1/2}\lambda^{-1/4m}| \alpha_1-\alpha_2| \sup_{\substack{|\alpha-\alpha_0|<\delta/2 \\ |u| < \delta}} \sup_{y \in \mathcal{Y}}  \left|f_{\theta_0(t)+u}^{\alpha}(y) \log\left(f_{\theta_0(t)+u}(y) \right) u_{\theta_0(t)+u}(y) \right| 
\\ & \quad \quad  + 2 c_0 C_n^{1/2}\lambda^{-1/4m} | \alpha_1-\alpha_2| \sup_{\substack{|\alpha-\alpha_0|<\delta/2 \\ |u| < \delta}} \sup_{y \in \mathcal{Y}} \left| f_{\theta_0(t)+u}^{\alpha}(y) u_{\theta_0(t)+u}(y) \right| 
\\ & \quad \leq c_0|\alpha_1-\alpha_2|, 
\end{align*}
for some global constant $c_0$ and we also have used the fact that $C_n^{1/2} \lambda^{-1/4m} \to 0$, hence $C_n^{1/2} \lambda^{-1/4m}$ is bounded.  A similar bound holds for the third term on the RHS of \eqref{eq:A15}. Putting everything together, we have
\begin{align*}
\sup_{t\in [0,1], y \in \mathcal{Y}}\left|h_{g_1, \alpha_1}(t,y) - h_{g_2,\alpha_2}(t,y)\right| &\leq c_0C_n^{1/2}\|g_1-g_2 \|_{\infty}+c_0|\alpha_1-\alpha_2|,
\end{align*}
which implies that, for every $\epsilon>0$,
\begin{align}
\label{eq:A16}
\mathcal{H}_{\infty}(\epsilon, \{h_{g, \alpha},\ g \in \mathcal{B}_D, \ \alpha \in \mathcal{V}_{\alpha_0} \}) \leq \mathcal{H}_{\infty}(c_0 \epsilon/ C_n^{1/2}, \mathcal{B}_D) + \mathcal{H}(c_0 \epsilon, \mathcal{V}_{\alpha_0}),
\end{align}
where  $\mathcal{H}_{\infty}(c_0 \epsilon/ C_n^{1/2}, \mathcal{B}_D)$ denotes the $c_0\epsilon/C_n^{1/2}$-entropy in the sup-norm of $\mathcal{B}_D \subset \mathcal{W}^{m,2}([0,1])$ and $\mathcal{H}(c_0 \epsilon, \mathcal{V}_{\alpha_0})$ denotes the $c_0\epsilon$-entropy of the interval $\mathcal{V}_{\alpha_0} \subset \mathbbm{R}$. 

We now bound the entropies on the RHS of \eqref{eq:A16}. By \eqref{eq:A8} we immediately have
\begin{align*}
\mathcal{H}_{\infty}(c_0\epsilon/C_n^{1/2}, \mathcal{B}_D) \leq c_0 C_n^{1/2m} (\lambda^{1/2}\epsilon)^{-1/{m}}.
\end{align*}
Moreover, by Lemma 2.5 of \citet{van de Geer:2000}, 
\begin{align*}
H(c_0\epsilon, \mathcal{V}_{\alpha_0}) \leq \log\left(\frac{c_0}{\epsilon}+1\right).
\end{align*}
For all $x\geq 0$ and natural numbers $m \in \mathbbm{N}$ we have $\log(1+x^m) \leq \log (m!) + x$. Hence, for all $x \geq \log (m!)$, $\log(1+x^m) \leq 2 x$. Applying this inequality with $x = (c_0/\epsilon)^{1/m}$ and sufficiently small $\epsilon>0$, yields
\begin{align*}
H(c_0\epsilon, \mathcal{V}_{\alpha_0}) \leq 2 \frac{c_0^{1/m}}{\epsilon^{1/m}} \leq c_0 \frac{C_n^{1/2m}}{\lambda^{1/2m} \epsilon^{1/m}},
\end{align*}
for some $c_0>0$ where we have also used $C_n > \lambda$ and consequently $C_n/\lambda>1$. Returning to \eqref{eq:A16}, we now have
\begin{align*}
\mathcal{H}_{B,K}(\epsilon, \{h_{g, \alpha},\ g \in \mathcal{B}_D, \ \alpha \in \mathcal{V}_{\alpha_0} \}, \bar{P}) \leq \mathcal{H}_{\infty}(\epsilon, \{h_{g, \alpha},\ g \in \mathcal{B}_D, \ \alpha \in \mathcal{V}_{\alpha_0} \}) \leq c_0 \frac{C_n^{1/2m}}{\lambda^{1/2m} \epsilon^{1/m}},
\end{align*}
for some $c_0>0$ and all small $\epsilon>0$.  Remembering that by our limit assumptions $C_n \lambda^{-1/4m} \to 0$, the bracketing integral of the class of functions $\{h_{g, \alpha},\ g \in \mathcal{B}_D, \ \alpha \in \mathcal{V}_{\alpha_0} \}$ from $0$ to $R = c_0 C_n \lambda^{-1/(4m)}$ may be bounded by
\begin{align*}
\int_{0}^{c_0 C_n \lambda^{-1/4m}} \mathcal{H}_{B,K}^{1/2}(u, \{h_{g, \alpha},\ g \in \mathcal{B}_D, \ \alpha \in \mathcal{V}_{\alpha_0} \}, \bar{P})du \leq c_0 C_n^{1-1/4m} \lambda^{1/8m^2} \lambda^{-1/2m},
\end{align*}
for all large $n$, as $R \to 0$, by our limit assumptions. Fix $\epsilon^{\prime}>0$ and take $a = \epsilon^{\prime} n^{1/2} C_n$ in Theorem 5.11 of \citet{van de Geer:2000}. We need to check the conditions
\begin{align*}
\epsilon^{\prime} n^{1/2} C_n & \geq c_0 C_n^{1-1/4m} \lambda^{1/8m^2} \lambda^{-1/2m} \\ \epsilon^{\prime} n^{1/2} C_n & \geq c_0 \lambda^{-1/4m} C_n,
\end{align*}
for a sufficiently large positive constant $c_0$. The latter condition is satisfied whenever $n^{1/2} \lambda^{1/4m} \to \infty$, equivalently, $n^{2m} \lambda \to \infty$. Our assumption $n^{m} \lambda \to \infty$ as $n \to \infty$ ensures that this condition is satisfied. The former condition is also satisfied under this limit assumption. Therefore, setting $C_1 = C_0$ in Theorem 5.11 of \citet{van de Geer:2000} yields
\begin{align*}
\Pr\left[ \sup_{g \in \mathcal{B}_{D},   \alpha \in \mathcal{V}_{\alpha_0} } |I_1(g, \alpha) - \mathbb{E}\{I_1(g, \alpha) \}| \geq \epsilon^{\prime} C_n  \right] & = \Pr\left[ \sup_{g \in \mathcal{B}_D,  \alpha \in \mathcal{V}_{\alpha_0} } |v_n(h_{g,\alpha})| \geq \epsilon^{\prime} n^{1/2} C_n \right]
\\ & \leq  c_0 \exp \left[ - c_0 |\epsilon^{\prime}|^2 n   \lambda^{1/2m} \right]
\end{align*}
for all large $n$. The exponential tends to zero for every $\epsilon^{\prime}$, hence  we have established \eqref{eq:A4} and the result of the theorem follows.

\end{proof}

\begin{proof}[Proof of Corollary~1]

The proof may be deduced from \citet[Chapter~13, Lemma~2.17]{Eg:2009}, which establishes the embedding
\begin{equation*}
\{\|f\|^2 + \lambda^{j/m} \|f^{(j)}\|^2  \}^{1/2} \leq c_0 \|f\|_{m , \lambda},
\end{equation*}
for all $j \leq m$ and  $f \in \mathcal{W}^{m, 2}([0,1])$ with $c_0$ depending only on $j$ and $m$. Since $\mathcal{W}^{m, 2}([0,1])$ is a vector space, Theorem 1 now implies that for any $1 \leq j \leq m$
\begin{equation*}
n^{-j/(2m+1)} \|\widehat{g}_n^{(j)} - g_0^{(j)}\| \leq c_0 \|\widehat{g}_n-g_0\|_{m, \lambda} = O_P(n^{-m/(2m+1)}),
\end{equation*}
for $\lambda \asymp n^{-2m/(2m+1)}$. The result follows.
\end{proof}

\section{Proofs of Proposition 1, Theorem 2 and Corollary 2}

For the proofs of this section we will use the analogue of Lemma 2 for spline functions. The proof of Lemma 3 below is given in \citet[Lemma 6.1]{Zhou:1998}.
\begin{lemma}
\label{lem:3}
Assume (B5)--(B7). Then, there exist constants $0<c_1<c_2<\infty$ (independent of $n$ and $K$) such that for any $f \in S_{K}^p([0,1])$ and $n \geq n_0$,
\begin{align*}
c_1 \|f\|^2 \leq \frac{1}{n} \sum_{i=1}^n |f(t_i)|^2 \leq c_2 \|f\|^2.
\end{align*}
\end{lemma}

We first provide the proof of Proposition 1. For the proof we make use of the B-spline functions $B_{1,p}, \ldots, B_{K+p,p}$ supported by the $K$ interior knots $0<x_1< \ldots, x_K<1$ and their properties. See \citet{DB:2001} for a thorough treatment. For simplicity we drop the $p-$subscript and simply write $B_1, \ldots, B_{K+p}$.

\begin{proof}[Proof of Proposition~1] We only need to prove $\sup_{t \in [0,1]}|f(t)| \leq c_0 K^{1/2} \|f\|_{m,\lambda}$ for some positive and finite $c_0$, as for $p>m$ we have $S_{K}^p([0,1]) \subset \mathcal{W}^{m,2}([0,1])$ and by (7) in the main text for all $\lambda \in (0,1]$ there exists $c_0$ such that $\sup_{t \in [0,1]}|f(t)| \leq c_0 \lambda^{-1/4m} \|f\|_{m,\lambda}$.

For any $f \in S_{K}^p([0,1])$ write $f(t) = \sum_{j=1}^{K+p} f_j B_{j}(t)$. The Schwarz inequality yields
\begin{align*}
|f(t)|^2 & \leq \left\{\sum_{j=1}^{K+p} |f_j|^2 \right\}\left\{\sum_{j=1}^{K+p} |B_{j}(t)|^2 \right\}  \\ & \leq \left\{\sum_{j=1}^{K+p} |f_j|^2 \right\}\left\{ \sup_{t \in [0,1]}\max_{1 \leq j \leq {K+p}}|B_j(t)| \sum_{j=1}^{K+p} |B_{j}(t)| \right\}
\\ & \leq \sum_{j=1}^{K+p} |f_j|^2,
\end{align*}
where to derive the last inequality we have used the facts $0\leq B_j(t) \leq 1$ and  $\sum_{j=1} B_j(t) = 1$ for every $t \in [0,1]$, see \citet[p. 96]{DB:2001}. By Lemma 6.1 in \citet{Zhou:1998}, there exists a positive constant $c_1$ depending only on $Q$ such that
\begin{align*}
\sum_{j=1}^{K+p} |f_j|^2 \leq c_1 K \int |f(t)|^2 dQ(t) \leq c_0 K \int |f(t)|^2 dt \leq c_0 K\left( \|f\|^2 + \lambda \|f^{(m)}\|^2 \right),
\end{align*}
where the second inequality follows from (B8) and the third inequality from the positivity of $\lambda$ and the map $f \mapsto \|f^{(m)}\|^2$. Putting everything together, there exists a constant $c_0$ not depending on $f \in S_{K}^p([0,1])$, $\lambda$ or $K$ such that
\begin{align*}
\sup_{t \in [0,1]} |f(t)| \leq c_0 K^{1/2} \|f\|_{m,\lambda},
\end{align*}
as asserted.

\end{proof}

For the proof of Theorem~2 we will need the following approximation lemma.

\begin{lemma}
\label{lem:4}
For each $f \in \mathcal{C}^{j}([0,1])$ there exists a spline function $s_f$ of order $p$ with $p >j$ such that
\begin{equation*}
\sup_{t \in [0,1]} | f(t) - s_f(t) | \leq c_0 |\mathbf{x}|^{j} \sup_{|t-y|<\mathbf{x}} |f^{(j)}(t)-f^{(j)}(y)|,
\end{equation*}
where $x_i$ are the knots,  $\mathbf{x} = \max_i|x_i - x_{i-1}|$ is the maximum distance of adjacent knots and the constant $c_0$ depends only on $p$ and $j$.
\end{lemma}
\begin{proof}
See \citet[pp. 145--149]{DB:2001}.
\end{proof}

We now turn to the proof of Theorem~2.

\begin{proof}[Proof of Theorem~2]

As previously, we denote the objective function in (6) of the main text with $L_n(g, \widehat{\alpha}_n)$, that is, for every $g \in S_{K}^p([0,1])$,
\begin{align*}
L_n(g, \widehat{\alpha}_n)= n^{-1} \sum_{i=1}^n l_{\widehat{\alpha}_n}(Y_i,g(t_i)) + \lambda\|g^{(m)}\|^2,
\end{align*}
and let $g_0 \in \mathcal{C}^{j}([0,1])$ denote the true function. Furthermore, let $s_{g_0}$ denote the spline approximation to $g_0$ constructed with the help of Lemma~\ref{lem:4}. Lemma~\ref{lem:4} and assumptions (B6)--(B7) imply that
\begin{align}
\label{eq:A17}
\sup_{t \in [0,1]}|g_0(t) - s_{g_0}(t)| = O(K^{-j}).
\end{align}
Hence, we may write $s_{g_0}(t_i) = g_0(t_i) + R(t_i)$ with $R(t_i) =   s_{g_0}(t_i) - g_0(t_i)$. As in the proof of Theorem~1, the theorem will be proven if establish that for every $\epsilon>0$ there exists a sufficiently large $ D = D_{\epsilon}$ such that
\begin{align}
\label{eq:A18}
\liminf_{n \to \infty}\Pr\left[ \inf_{g \in S_{K}^p([0,1]):\|g\|_{m,\lambda} = D} L_n(s_{g_0} + C_n^{1/2}g, \widehat{\alpha}_n) > L_n(s_{g_0}, \widehat{\alpha}_n), |\widehat{\alpha}_n-\alpha_0|<\delta/2 \right] \geq 1-\epsilon/2,
\end{align}
where $C_n = n^{-1} \min\{K, \lambda^{-1/2m}\} + \min\{\lambda^2 K^{2m},\lambda\} + K^{-2j}$, as an application of Lemma~\ref{lem:1} would yield the existence of a $\widehat{g}_n \in S_{K}^p([0,1])$ such that $\|\widehat{g}_n-s_{g_0}\|_{m,\lambda} \leq D C_n^{1/2}$ with probability at least $1-\epsilon$. Note that $S_{K}^p([0,1])$ is a finite-dimensional Hilbert space under $\langle \cdot, \cdot \rangle_{m,\lambda}$, hence weak and strong continuity as well as weak and strong convergence are equivalent. Now, by \eqref{eq:A17},
\begin{align*}
\|\widehat{g}_n-g_0\|^2 & \leq 2\|\widehat{g}_n-s_{g_0}\|^2 + 2\|s_{g_0}-g_0\|^2 \leq 2D C_n + 2D^{\prime}K^{-2j} \leq  4DC_n,
\end{align*}
for all large $D$ with probability at least $1-\epsilon$. Thus, the result of Theorem~2 follows upon proving \eqref{eq:A18}.

To establish \eqref{eq:A18} we decompose  $L_n(s_{g_0} + C_n^{1/2}g, , \widehat{\alpha}_n) - L_n(s_{g_0}, , \widehat{\alpha}_n)$ as follows:
\begin{align*}
L_n(s_{g_0} + C_n^{1/2}g, \widehat{\alpha}_n) - L_n(s_{g_0}, \widehat{\alpha}_n) & = \frac{1}{n} \sum_{i=1}^n l_{\widehat{\alpha}_n}(Y_i, s_{g_0}(t_i) +C_n^{1/2}g(t_i)) - \frac{1}{n} \sum_{i=1}^n l_{\widehat{\alpha}_n}(Y_i, s_{g_0}(t_i))
\\ & \quad +2 \lambda C_n^{1/2} \langle s_{g_0}^{(m)}, g^{(m)} \rangle + \lambda C_n \|g^{(m)}\|^2
\\ & =\frac{1}{n} \sum_{i=1}^n \int_{R(t_i)}^{R(t_i)+ C_n^{1/2}g(t_i)} \{l_{\widehat{\alpha}_n}^{\prime}(Y_i, g_0(t_i)+u) - l_{\widehat{\alpha}_n}^{\prime}(Y_i, g_0(t_i))\} du
\\ & \quad + \frac{C_n^{1/2}}{n} \sum_{i=1}^n g(t_i) l_{\widehat{\alpha}_n}^{\prime}(Y_i, g_0(t_i)) + 2 \lambda C_n^{1/2}  \langle s_{g_0}^{(m)}, g^{(m)} \rangle
\\ &\quad + \lambda C_n\|g^{(m)}\|^2
\\ & = I_1(g, \widehat{\alpha}_n) + I_2(g, \widehat{\alpha}_n) + I_3(g),
\end{align*}
say, with
\begin{align*}
I_1(g, \widehat{\alpha}_n) &:= \frac{1}{n} \sum_{i=1}^n \int_{R(t_i)}^{R(t_i)+C_n^{1/2}g(t_i)} \{l_{\widehat{\alpha}_n}^{\prime}(Y_i, g_0(t_i)+u) - l_{\widehat{\alpha}_n}^{\prime}(Y_i, g_0(t_i))\} du + \lambda C_n \|g^{(m)}\|^2
\\ I_2(g, \widehat{\alpha}_n) &:=  \frac{C_n^{1/2}}{n}  \sum_{i=1}^n g(t_i) l_{\widehat{\alpha}_n}^{\prime}(Y_i, g_0(t_i))
\\ I_3(g) &=  2 \lambda C_n^{1/2}  \langle s_{g_0}^{(m)}, g^{(m)} \rangle.
\end{align*}
We will show that
\begin{align}
\inf_{ \substack{ g \in S_{K}^p([0,1]):\|g\|_{m, \lambda} = D \\ |\alpha-\alpha_0| < \delta/2 }} \mathbb{E}\{I_1(g, \alpha) \} & \geq c_0 D^2 C_n - c_1 D C_n \label{eq:A19} \\
\sup_{\substack{g \in S_{K}^p([0,1]):\|g\|_{m, \lambda} \leq D \\  |\alpha-\alpha_0| < \delta/2}} |I_1(g, \alpha)-\mathbb{E}\{I_1(g, \alpha)\}| & = o_P(1) C_n \label{eq:A20}
\\ \sup_{g \in S_{K}^p([0,1]):\|g\|_{m, \lambda} \leq D} |I_2(g, \widehat{\alpha}_n)| &= O_P(1) D C_n \label{eq:A21} \\
\sup_{g \in S_{K}^p([0,1]):\|g\|_{m, \lambda} \leq D} |I_3(g)| & = O(1) D C_n \label{eq:A22},
\end{align}
for a strictly positive $c_0$. These are sufficient for Theorem~2 to hold, as, for large enough satisfying $D\geq 1$, the infimum of  $\mathbb{E}\{I_1(g, \widehat{\alpha}_n) \}$ will be positive and dominate all other terms in the decomposition guaranteeing \eqref{eq:A18}.

Beginning with \eqref{eq:A22}, the Schwarz inequality immediately yields
\begin{align*}
\sup_{g \in S_{K}^p([0,1]):\|g\|_{m, \lambda} \leq D}|I_3(g)| \leq 2 \lambda C_n^{1/2}\| s_{g_0}^{(m)}\| \|g^{(m)}\| \leq c_0 \lambda C_n^{1/2} \sup_{g \in S_{K}^p([0,1]):\|g\|_{m, \lambda} \leq D}\|g^{(m)}\|.
\end{align*}
Here, we have used the boundedness of $||s_{g_0}^{(m)}||$, see \citet[p. 155]{DB:2001}. By definition of $\|\cdot\|_{m,\lambda}$, we have $\lambda^{1/2}\|g^{(m)}\| \leq \|g\|_{m,\lambda}$. At the same time, for every $g \in S_{K}^p([0,1])$ with $g = \sum_j g_j B_{j}$,
\begin{align*}
\|g^{(m)}\|^2 = \|D^m g\|^2 \leq c_0 K^{2m-1} \sum_{j=1}^{K+p} |g_j|^2 \leq c_0 K^{2m} \|g\|_{m,\lambda}^2,
\end{align*}
where $D^m$ denotes the $m$th order differentiation operator on $S_{K}^p$, the second inequality follows from Lemma 5.2 of \citet{Cardot:2002} and the third inequality follows as in the proof of Lemma~3. Combining these two bounds we find
\begin{align*}
\sup_{g \in S_{K}^p([0,1]):\|g\|_{m, \lambda} \leq D}|I_3(g)| \leq c_0 D \lambda C_n^{1/2} \min\{ K^m, \lambda^{-1/2} \}  = c_0 D C_n^{1/2} \min\{ \lambda K^m, \lambda^{1/2} \} \leq c_0 D C_n,
\end{align*}
as, by definition of $C_n$,  $\min\{ \lambda K^m, \lambda^{1/2} \} \leq C_n^{1/2}$. 

We now prove \eqref{eq:A21}. As in the proof of Theorem 1, the crucial quantity is the empirical entropy $H(\epsilon, \{g \in S_{K}^p([0,1]): \|g\|_{m,\lambda} \leq D \}, Q_n)$.  Denote for simplicity $\mathcal{B}_D = \{g \in S_{K}^p([0,1]): \|g\|_{m,\lambda} \leq D\}$. Since, $S_{K}^p([0,1]) \subset \mathcal{W}^{m,2}([0,1])$ for $p >m$, by \eqref{eq:A8}, 
\begin{align}
\label{eq:A23}
H(\epsilon, \mathcal{B}_D, Q_n) \leq C D^{1/m} \lambda^{-1/2m} \epsilon^{-1/m}.
\end{align}
At the same time for any $f,g \in S_{K}^p([0,1])$, by Lemma\ref{lem:3}, we have $\|f-g\|_{Q_n} \leq c_2 \|f-g\|$.
It follows that $H(\epsilon, \mathcal{B}_D, Q_n) \leq H(c_0 \epsilon, \mathcal{B}_D)$ where the latter denotes the entropy with respect to the $\mathcal{L}^2([0,1])$ distance. Notice further that $\mathcal{B}_D \subset \{f \in S_{K}^p([0,1]): \|f\| \leq D\}$. Hence, by Corollary 2.6 in \citet{van de Geer:2000}, we have
\begin{align}
\label{eq:A24}
H(\epsilon, \mathcal{B}_D, Q_n) \leq H(c_0 \epsilon, \mathcal{B}_D) \leq H(c_0 \epsilon, \{f \in S_{K}^p([0,1]): \|f\| \leq D\}) \leq (K+p) \log \left( \frac{4Dc_0}{\epsilon} + 1 \right).
\end{align}
Taking square roots and integrating the bounds in \eqref{eq:A23} and \eqref{eq:A24}, we find that for all large $n$ there exists a $c_0>0$ independent of $n$ such that
\begin{align}
\label{eq:A25}
\int_{0}^{c_0 D} H^{1/2}(u, \mathcal{B}_D, Q_n) du \leq c_0 D \min\{K^{1/2}, \lambda^{-1/4m}\},
\end{align}
where we have used the fact that
\begin{align*}
\int_{0}^{c_0 D} \log^{1/2} \left( \frac{4Dc_0}{u} + 1 \right) du = 4 c_0 D \int_{0}^{1/4} \log^{1/2} \left( \frac{1}{u} + 1 \right) du 
\end{align*}and the latter integral is independent of $n$ and finite. As in the proof of Theorem 1, \eqref{eq:A25} and Corollary 8.3 of \citet{van de Geer:2000} now imply that for any $\alpha \in (\alpha_0-\delta, \alpha_0+\delta)$,
\begin{align*}
\Pr\left( \sup_{g \in \mathcal{B}_D} \left| \frac{1}{n} \sum_{i=1}^n g(t_i) l^{\prime}_{\alpha}(Y_i,g_0(t_i)) \right| \geq  \frac{c_0 D \min\{K^{1/2}, \lambda^{-1/4m}}{2 n^{1/2}} \right) \leq c_0 \exp\left[ -c_0 \min\{K, \lambda^{-1/2m}\} \right],
\end{align*}
so that yet again by our assumptions and the union bound we find
\begin{align*}
\Pr & \left( \sup_{g \in \mathcal{B}_D, |\alpha-\alpha_0|<\delta/2} \left| \frac{1}{n} \sum_{i=1}^n g(t_i) l^{\prime}_{\alpha}(Y_i,g_0(t_i)) \right|  \geq  \frac{c_0 D \min\{K^{1/2}, \lambda^{-1/4m}\}}{2 n^{1/2}} \right) \\ & \leq c_0 n^{1/2} \max\{K^{-1/2}, \lambda^{1/4m} \} \exp\left[ -c_0 \min\{K, \lambda^{-1/2m}\} \right],
\end{align*}
and the RHS of this inequality tends to zero, by our limit assumptions. Hence,
\begin{align*}
\sup_{g \in S_{K}^p([0,1]):\|g\|_{m, \lambda} \leq D} |I_2(g, \widehat{\alpha}_n)| & \leq \sup_{g \in S_{K}^p([0,1]):\|g\|_{m, \lambda} \leq D, |\alpha - \alpha_0| < \delta/2} |I_2(g, \alpha)|  
\\ & = D C_n^{1/2} O_P(1) \frac{\min\{K^{1/2}, \lambda^{-1/4m}\}}{n^{1/2}}  \\ & = O_P(1) D C_n, 
\end{align*}
which yields \eqref{eq:A21}.

We now establish a uniform lower bound on $\mathbb{E}\{I_1(g, \alpha)\}$, as required in \eqref{eq:A19}. For this, first note that $\sup_{t \in [0,1]} |R(t)| = o(1)$ as $n \to \infty$, by \eqref{eq:A16} and (B6). Furthermore, Proposition 1 in the main text yields the existence of a reproducing kernel $\mathcal{R}_{m,K,\lambda}: [0,1]^2 \to \mathbbm{R}$  such that, for every $g \in S_{K}^p([0,1])$, $g(t) = \langle g, \mathcal{R}_{m,K,\lambda}(t, \cdot)  \rangle_{m,\lambda}$ and
\begin{align}
\label{eq:A26}
\sup_{t \in [0,1]} \| \mathcal{R}_{m,K,\lambda}(t, \cdot)\|_{m,K,\lambda} \leq c_0 \min\{\lambda^{-1/(4m)}, K^{1/2} \},
\end{align}
for some universal constant $c_0$, not depending on $K$ or $\lambda$.  Using these two properties and the Schwarz inequality we now see that
\begin{align}
\label{eq:A27}
C_n^{1/2} \sup_{g \in S_{K}^p([0,1]): \|g\|_{m,\lambda} \leq D} \|g\|_{\infty} \leq c_0 D C_n^{1/2} \min\{\lambda^{-1/(4m)}, K^{1/2} \} = o(1),
\end{align}
as $n \to \infty$, by our limit assumptions. Now, for any $\alpha \in (\alpha_0-\delta/2, \alpha_0+\delta/2)$  a first order Taylor expansion about zero shows the existence of an $|s_i| \leq 1$ such that
\begin{align}
\mathbb{E}\{l_{\alpha}^{\prime}(Y_i, g_0(t_i)+u)\} du & =  \mathbb{E}\{l_{\alpha}^{\prime \prime}(Y_i, g_0(t_i))\} u   + \left[\mathbb{E}\{l_{\alpha}^{\prime \prime}(Y_i, g_0(t_i)+ s_i u) - \mathbb{E}\{l_{\alpha}^{\prime \prime}(Y_i, g_0(t_i))\right]u
\label{eq:A28},
\end{align}
as, by Fisher consistency, $\mathbb{E}\{l_{\alpha}^{\prime}(Y_i, g_0(t_i))\} = 0$. Now, in the notation of (A4),
\begin{align*}
\left|\mathbb{E}\{l_{\alpha}^{\prime \prime}(Y_i, g_0(t_i)+ s_i u) - \mathbb{E}\{l_{\alpha}^{\prime \prime}(Y_i, g_0(t_i))\right| & \leq \left|m_{t_i}(s_i u, \alpha)-m_{t_i}(0, \alpha)\right|
\\& \quad + \left(1+\frac{1}{\alpha}\right) \mathbb{E} \left\{\left| n_{t_i}(s_i u, \alpha, Y_i) - n_{t_i}(0, \alpha, Y_i)  \right| \right\}
\\ & = o(1),
\end{align*}
as $u \to 0$, uniformly in $i$ and $n$ and $\alpha \in (\alpha_0-\delta/2, \alpha_0+\delta/2)$, by equicontinuity and dominated convergence. Using \eqref{eq:A28} inside the integral below we get
\begin{align*}
\int_{R(t_i)}^{R(t_i)+ C_n^{1/2}g(t_i)} \mathbb{E}\{l_{\alpha}^{\prime}(Y_i, g_0(t_i)+u))\} du &= \int_{R(t_i)}^{R(t_i)+ C_n^{1/2}g(t_i)} \mathbb{E}\{l_{\alpha}^{\prime \prime}(Y_i, g_0(t_i))\}u(1  + o(1)) du
\\ & =  \mathbb{E}\{l_{\alpha}^{\prime \prime}(Y_i, g_0(t_i))\} \left[ 2^{-1} C_n |g(t_i)|^2 + R(t_i) C_n^{1/2}g(t_i) \right](1+o(1)).
\end{align*}
Notice that the approximation is valid, since the domains of integration tend to zero as $n \to \infty$. Noting now that $\mathbb{E}\{l_{\alpha}^{\prime \prime}(Y_i, g_0(t_i))\} = \mathbb{E} \{ f_{\theta_{0,i}}(Y_i) |u_{\theta_{0,i}}(Y_i)|^2 \}$, (A5) reveals that, for all large $n$,
\begin{align*}
\int_{R(t_i)}^{R(t_i)+ C_n^{1/2}g(t_i)} \mathbb{E}\{l_{\alpha}^{\prime}(Y_i, g_0(t_i)+u))\} du \geq c_0 C_n |g(t_i)|^2 - c_1 C_n^{1/2}|R(t_i)| |g(t_i)|,
\end{align*}	
for strictly positive $c_0$ and $c_1$ that, by (A5), does not depend on $i$ and $\alpha \in (\alpha_0-\delta/2, \alpha_0+\delta/2)$. Averaging and approximating the sum from below with the help of Lemma~\ref{lem:3},
\begin{align*}
\inf_{|\alpha-\alpha_0|<\delta}\mathbb{E}\{I_1(g,\alpha)\} & \geq c_0 \frac{C_n}{n} \sum_{i = 1}^n |g(t_i)|^2 - c_1 \frac{C_n^{1/2}}{n} \sum_{i=1}^n |R(t_i)| |g(t_i)| + \lambda \|g^{(m)} \|^2
 \\ & \geq  c_0 C_n \|g\|_{m,\lambda}^2 - c_1 \frac{C_n^{1/2}}{n} \sum_{i=1}^n |R(t_i)| |g(t_i)|.
\end{align*}
Furthermore, using Lemma~\ref{lem:4} along with (B6) in order to bound $R(t_i)$, we get
\begin{align*}
\left|C_n^{1/2}n^{-1} \sum_{i=1}^n R(t_i) g(t_i)\right| & \leq c_1 C_n^{1/2} K^{-j} n^{-1}\sum_{i=1}^n|g(t_i)|
\\ & \leq c_1 C_n^{1/2} K^{-j} \left\{n^{-1} \sum_{i=1}^n |g(t_i)|^2 \right\}^{1/2}
\\ & \leq c_1 C_n \|g\|_{m,\lambda},
\end{align*}
for some $0<c_1<\infty$. To derive the last inequality we have used $K^{-j} \leq C_n^{1/2}$, Lemma~\ref{lem:3} and the inequality $\|g\| \leq \|g\|_{m,\lambda}$. Combining the above, we find
\begin{align*}
\inf_{\substack{ g \in S_{K}^p([0,1]): \|g\|_{m,\lambda} = D \\ |\alpha-\alpha|<\delta/2} } \mathbb{E}\{I_1(g, \alpha)\} &\geq c_0 D^2 C_n - c_1 D C_n,
\end{align*}
which is precisely \eqref{eq:A19}.

To complete the proof we now show \eqref{eq:A20} and for this we largely adopt the notation in the proof of Theorem~1. In this notation we may write
\begin{align*}
I_1(g, \alpha) - \mathbb{E}\{I_1(g, \alpha) \} = \int h_{g, \alpha} d(P_n - \bar{P}) = n^{-1/2} v_n(h_{g, \alpha}),
\end{align*}
where $v_n(\cdot)$ denotes the empirical process and $h_{g, \alpha}$ is the function $[0,1] \times \mathbbm{R} \to \mathbbm{R}$ given by
\begin{align*}
h_{g,\alpha}(t,y) := \int_{R(t)}^{R(t) + C_n^{1/2}g(t)}\{l^{\prime}_{\alpha}(y, g_0(t)+u) - l^{\prime}_{\alpha}(y, g_0(t)) \} du,
\end{align*}
for each $g \in \mathcal{B}_{D} := \{f \in S_{K}^p	([0,1]): \|f\|_{m,\lambda} \leq D\}$ and $\alpha \in \mathcal{V}_{\alpha_0} := (\alpha_0-\delta/2, \alpha_0+\delta/2)$. The proof is based on Theorem 5.11 of \citet{van de Geer:2000}; we avoid repetitions and provide only its most important elements, namely a uniform bound on the class of functions $\{h_{g,\alpha}, g \in \mathcal{B}_{D}, \alpha \in \mathcal{V}_{\alpha_0} \}$, a uniform bound on its $L^2(\bar{P})$-norm, the Lipschitz constants and a bound on the covering number. In particular, under our assumptions with arguments similar as in the proof of Theorem~1, it is easy to show that
\begin{align}
\label{eq:A29}
\sup_{\substack{g \in \mathcal{B}_D, \alpha \in \mathcal{V}_{\alpha_0}}} \|h_{g, \alpha}\|_{\infty} & \leq c_0 \min\{K, \lambda^{-1/(2m)} \} C_n,
\end{align}
for all large $n$, as $\sup_{t \in [0,1]}|R(t)| =O(K^{-j}) = O(C_n^{1/2})$ and $K \to \infty$. Secondly,
\begin{align}
\label{eq:A30}
\sup_{\substack{g \in \mathcal{B}_D, \alpha \in \mathcal{V}_{\alpha_0}}}\int |h_{g, \alpha}|^2 d \bar{P}  \leq c_0 \min\{K, \lambda^{-1/(2m)} \} C_n^2.
\end{align}
In addition, for any $(g_1, \alpha_1), (g_2, \alpha_2) \in \mathcal{B}_D \times \mathcal{V}_{\alpha_0}$  we have
\begin{align}
\label{eq:A31}
\left|h_{g_1, \alpha_1}(t,y) - h_{g_2,\alpha_2}(t,y)\right| & \leq c_0 C_n^{1/2} |g_1(t)-g_2(t)| + c_0 |\alpha_1-\alpha_2| + c_0 C_n^{1/2} K^{1/2}|\alpha_1-\alpha_2| \nonumber
\\ & \leq c_0 C_n^{1/2} K^{1/2} \|g_1-g_2\| + c_0 C_n^{1/2} K^{1/2}|\alpha_1-\alpha_2| \nonumber
\\ & \leq c_0 \|g_1-g_2\| + c_0 |\alpha_1-\alpha_2|
\end{align}
where the second-to-last inequality follows as in the proof of Proposition~1 and the last inequality from the fact that $C_n^{1/2} K^{1/2} \to 0$, by our limit assumptions. By Lemma 2.5 of \citet{van de Geer:2000}, it follows that
\begin{align}
\label{eq:A32}
\mathcal{H}_{\infty}(\epsilon, \{h_{g,\alpha}, g\in \mathcal{B}_D, \alpha \in \mathcal{V}_{\alpha_0}\}) \leq (K+p+1) \log\left(\frac{c_0}{\epsilon}+1\right).
\end{align}

By theorem 5.11 of \citet{van de Geer:2000} and under our limit assumptions, \eqref{eq:A29}--\eqref{eq:A32} imply that for every $\epsilon^{\prime}>0$,
\begin{align*}
\Pr\left[ \sup_{\substack{g \in S_{K}^p([0,1]):\|g\|_{m, \lambda} \leq D \\  |\alpha-\alpha_0| < \delta/2}} |I_1(g, \alpha)-\mathbb{E}\{I_1(g, \alpha)\}| \geq \epsilon^{\prime} C_n  \right] & = \Pr\left[ \sup_{g \in \mathcal{B}_D, \alpha \in \mathcal{V}_{\alpha_0}} |v_n(h_{g, \alpha})| \geq \epsilon^{\prime} n^{1/2} C_n \right]
\\ & \leq  c_0 \exp \left[ - c_0 |\epsilon^{\prime}|^2 n/K \right]
\end{align*}
for all large $n$. The result follows, as, by (B6),  $n/K \to \infty$ for $n \to \infty$. The proof is complete.
\end{proof}

Finally, we provide the proof of Corollary~2.

\begin{proof}[Proof of Corollary~2] For $g_0 \in \mathcal{C}^m([0,1])$,  inspection of the proof of Theorem~2 reveals that we have actually shown the stronger
\begin{align*}
\|\widehat{g}_n-s_{g_0}\|_{m,\lambda}^2 = O_P(C_n),
\end{align*}
for $C_n = n^{-1} \min\{K, \lambda^{-1/2m}\} + \min\{\lambda^2 K^{2m}, \lambda\} + K^{-2m}$. Thus, by the inequality $|x+y|^2 \leq 2x^2+2y^2$, the definition of $\| \cdot \|_{m,\lambda}$ and Lemma~\ref{lem:4},
\begin{align*}
\|\widehat{g}_n-g_0\|_{m,\lambda}^2 & \leq 2 \|\widehat{g}_n-s_{g_0}\|_{m,\lambda}^2 + 2 \|s_{g_0}-g_0\|_{m,\lambda}^2
\\ &  = O_P(C_n) + O(K^{-2m}) + 2\lambda\|s_{g_0}^{(m)} - g_0^{(m)}\|^2.
\end{align*}
By Theorem (26) in \citet[p. 155]{DB:2001}, $\|s_{g_0}^{(m)} - g_0^{(m)}\| = O(1)$. Moreover, for =  $K \asymp n^{\gamma}$ with $\gamma \geq 1/(2m+1)$ we have $C_n = n^{-1} \lambda^{-1/2m} + \lambda + K^{-2m}$, so that
\begin{align*}
\|\widehat{g}_n-g_0\|_{m,\lambda}^2 = O_P(C_n) + O(K^{-2m}) + O(\lambda) = O_P(C_n).
\end{align*}
With our choice of tuning parameters, $K $ and $\lambda$, $\|\widehat{g}_n-g_0\|_{m,\lambda}^2 = O(n^{-2m/(2m+1)})$. The result now follows exactly as in the proof of Corollary~1.

\end{proof}

\end{document}